\documentclass[12pt]{article}
\usepackage{amsmath,amsthm,amssymb}
\usepackage{graphicx,psfrag,epsf}
\usepackage{enumerate}
\usepackage{natbib}
\usepackage{url} 
\usepackage[toc]{appendix}
\usepackage{float}

\newcommand{\blind}{0}

\newtheorem{definition}{Definition}[section]
\newtheorem{assumption}{Assumption}[section]

\newtheorem{example}{Example}[section]
\newtheorem{theorem}{Theorem}[section]
\newtheorem{proposition}[theorem]{Proposition}
\newtheorem{lemma}[theorem]{Lemma}
\newtheorem{corollary}[theorem]{Corollary}
\newtheorem{remark}{Remark}

\usepackage{version}
\newcommand{\indep}{\rotatebox[origin=c]{90}{$\models$}}

\DeclareMathAlphabet{\mathcalligra}{T1}{calligra}{m}{n}

\begin{document}

\def\spacingset#1{\renewcommand{\baselinestretch}%
{#1}\small\normalsize} \spacingset{1}


\if0\blind 
{
  \title{\bf  Noncompliance in randomized control trials without exclusion restrictions}
  \author{Masayuki Sawada\thanks{
 I thank Edward Vytlacil, Yuichi
 Kitamura, Yusuke Narita, Mitsuru Igami, and Joseph
 Altonji for their guidance and support. I have benefited from comments particularly from Sukjin Han, Hiro
 Kasahara, Donhyuk Kim, Jinwook Kim, Yoshiaki Omori, Jeff Qiu, Pedro Sant'Anna, Kensuke
 Teshima, and seminars participants at Yale, 2018
 Kyoto Summer Workshop, Workshop on Advances in Econometrics 2019, Osaka School
 of International Public Policy, Hitotsubashi Institute of Economic
 Research, Kyoto
 Institute of Economic Research, Auburn University, Kansas State
 University, and Yokohama National University. All errors are my own.
 This paper was a dissertation chapter and previously circulated as a job market paper titled
 ``Identification and
 inference of post-treatment subgroup effects.'' This work is supported by Research Grants for Young Researchers, Hitotsubashi University.}\hspace{.2cm}\\
    Institute of Economic Research, Hitotsubashi University.\\
    m-sawada@ier.hit-u.ac.jp}
  \maketitle
} \fi

\if1\blind
{
  \bigskip
  \bigskip
  \bigskip
  \begin{center}
    {\LARGE\bf Noncompliance in randomized control trials without exclusion restrictions}
\end{center}
  \medskip
} \fi

\bigskip
\begin{abstract}
This study proposes a method to identify treatment effects without exclusion restrictions in randomized experiments with noncompliance. Exploiting a baseline survey commonly available in randomized experiments, I decompose the intention-to-treat effects conditional on the endogenous treatment status. I then identify these parameters to understand the effects of the assignment and treatment. The key assumption is that a baseline variable maintains rank orders similar to the control outcome. I also reveal that the change-in-changes strategy may work without repeated outcomes. Finally, I propose a new estimator that flexibly incorporates covariates and demonstrate its properties using two experimental studies. 
\end{abstract}

\noindent%
\vfill

\newpage
\spacingset{1.45} 


\section{Introduction}
Randomized assignment is a standard strategy to identify the average treatment effect under full compliance (i.e., when the assignment equals the associated treatment). By contrast, under noncompliance, an additional exclusion restriction is required to identify the associated treatment effects. Although an exclusion restriction requires that the assignment has no causal effect on the outcomes holding the treatment status fixed, this restriction may not always hold. In fact, the assignment itself may affect the outcomes for several reasons. 


For randomized control trials (RCTs), I propose a strategy to identify the effects of a binary assignment and binary treatment without imposing an exclusion restriction. Specifically, I consider the intention-to-treat (ITT) effects decomposed by the treatment status. There are two key parameters to consider. First, I consider the ITT conditional on the treated under the assignment (ITTTA). The ITTTA is the combined effect of both the assignment and the treatment. It equals the average treatment effect on the treated (ATT), when the exclusion restriction holds under one-sided noncompliance, which prohibits the units from taking up the treatment without the assignment. Second, I consider the ITT conditional on the nontreated under the assignment (ITTNA). The ITTNA is solely the assignment's own effect, often referred to as the assignment's direct effect.

To identify the heterogeneous effects of whether to employ the treatment, the baseline strategies are to either use the assignment as an instrumental variable (IV) or the conditionally ignorable assumption for the treatment. In my strategy, I require neither the exogenous treatment nor valid IVs. Instead, I develop a natural experimental approach in RCTs with their typical feature, that is, the baseline survey. I use a continuous variable from the baseline survey as a proxy for the control outcome. Specifically, I develop a rank imputation strategy with new estimators based on a rank similarity assumption that appears in \citeauthor{ChernozhukovHansen05}'s (\citeyear{ChernozhukovHansen05}) quantile IV model as well as in \citeauthor{AtheyImbens06}'s (\citeyear{AtheyImbens06}) change-in-changes (CiC) model. I apply the method to microcredit and cash transfer experimental studies and demonstrate that the direct effect of the randomized assignment for the nontreated, the ITTNA, is an attractive and important parameter for estimating the combined effect of the assignment and the treatment, the ITTTA.

This study makes three major contributions to the literature. First, I propose a natural experimental strategy for RCTs by exploiting their typical feature, the baseline survey. Specifically, the first contribution concerns the parameter, called principal stratification by \cite{FrangakisRubin2002}, that appears when exploring for the causal effect mechanism. Causal inference with an invalid instrument for the treatment is the leading context in principal stratification. For example, \cite{DeuchertHuber2017}, who examined the Vietnam draft lottery, documented that the exclusion restriction may be violated. Furthermore, most studies in this context exploit exclusion restrictions (\citealp{ImbensAngrist1994}; \citealp{AngristImbensRubin96}) or conditional independence (\citealp{HiranoImbensRubinZhou2000}).



As an alternative, several studies (\citealp{ZhangRubin2003}; \citealp{Rubin2006}; \citealp{FloresFloresLagunes2013}; \citealp{MealliPacini2013}) have proposed partial identification strategies exploiting monotone treatment response arguments (\citealp{Manski1997}; \citealp{ManskiPepper2000}). In response to the above baseline strategies, the natural experimental approach has also emerged as a substitute recently. For instance, \cite{DeuchertHuberSchelker2018} used a difference-in-differences (DiD) strategy to study the Vietnam war draft's impact on political preferences. Specifically, the analysis allowed for the draft itself to have a direct impact on the preferences. Additionally, \cite{HuberSchelkerStrittmatter2019} proposed a CiC strategy for the principal stratification and mediation analysis to study the impact of job training on mental health, using a randomized job training program in the United States called JOBS II. A unique feature of \cite{HuberSchelkerStrittmatter2019} appears in their results without valid randomization. Specifically, they applied their strategy without valid randomization when examining the effect of paid maternity leave on labor income, using the Swiss Labour Force Survey.


Thus, I contribute to this stream of the literature by offering a unique natural experimental strategy that exploits the baseline survey with a continuous proxy variable. I impose restrictions on the proxy variable and the control outcome in their latent rankings, which are uniform random variables representing their underlying percentiles. I assume that the conditional distributions of the latent rankings for the proxy variable and control outcome are identical regardless of whether taking up a treatment. This restriction is called the rank similarity assumption, under which taking up a treatment may be endogenous. Although the strategy is consistent with \cite{HuberSchelkerStrittmatter2019}, my aims and imposed restrictions differ in two ways. First, \cite{HuberSchelkerStrittmatter2019} required a structural model with a common latent rank variable unifying the different potential outcomes. Although the functional form is entirely flexible, it is a strong restriction in which a single unobserved variable governs all the potential outcome values during the same period. As I define the latent ranks as reduced-form representations for each potential outcome, I do not impose the structural model or the common latent rank variable. Second, to fully uncover the parameters, \cite{HuberSchelkerStrittmatter2019} imposed latent rank restrictions not only on the control outcomes against the proxy but also on the treatment outcomes against the proxy. While the proxy variable and control outcome, which are both free of the intervention, are similar, it is challenging to relate the treatment outcomes to the proxy in general. Nonetheless, these stronger restrictions allow them to identify the mediation effects and principal stratification without randomization. Furthermore, it is difficult to disentangle the restrictions on the control and treatment outcomes because both are unified in the common latent rank, and the restrictions are imposed on the common latent rank variable. In my strategy, I do not impose any restrictions on the treated outcome, given the potential outcome model under randomization for the ITTTA and ITTNA. Instead, I only impose restrictions on the control outcomes. As these weaker restrictions prohibit us from identifying the other parameters in \cite{HuberSchelkerStrittmatter2019}, I impose fewer necessary restrictions on the potential outcomes to identify the parameters relating to the principal strata. 

Overall, compared with \cite{HuberSchelkerStrittmatter2019}, who considered a richer mediation analysis under richer models and restrictions, the proposed strategy requires fewer model structures and less restrictive assumptions to identify the principal strata parameters. Therefore, this study complements \cite{HuberSchelkerStrittmatter2019} by offering another option for applied researchers to trade off the strength of the restrictions and richness of the parameters. In general, one can use my strategy for minimal key parameters and \citeauthor{HuberSchelkerStrittmatter2019}'s (\citeyear{HuberSchelkerStrittmatter2019}) if stronger restrictions are justified to study the causal effect mechanisms in depth.



As a second contribution, I find that the CiC strategy does not necessarily require repeated outcome measures for point identification. I demonstrate that the CiC strategy with binary or discrete outcomes can be used to identify the treatment effect when there is a continuous proxy variable in the baseline. My approach again differs from \cite{HuberSchelkerStrittmatter2019}, who considered an identification with continuous repeated outcomes. Many rank imputation approaches involve rank similarity or a stronger version, rank invariance (\citealp{JuhnMurphyBrooks93}; \citealp{DiNardoFortinLemiux96}; \citealp{AltonjiBlank99}; \citealp{ChernozhukovHansen05}; \citealp{MachadoMata05}). An alternative copula restriction also reveals the identification of the conditional quantile treatment effects (\citealt{CallawayLiOka17}), but the rank similarity approach is dominant. More recently, \cite{Han2018} used the rank similarity assumption to identify dynamic treatment effects with nonseparable models and \cite{deChaisemartinDHaultfoeuille2017} studied the fuzzy design of CiC models under the assumption of an exclusion restriction.

Among these rank similarity-based approaches,  \citeauthor{AtheyImbens06}' (\citeyear{AtheyImbens06}) CiC strategy is a natural experimental approach. As a generalization of the DiD model, a CiC strategy considers repeated outcome measures. For example, the underlying motivation in \citeauthor{AtheyImbens06}' (\citeyear{AtheyImbens06}) CiC strategy was to examine the issue in the standard DiD model, specifically because discrete outcomes may be predicted outside their supports. However, a CiC strategy with discrete repeated outcomes does not identify the parameters without stronger restrictions. Thus, I demonstrate that a CiC strategy may work without repeated outcomes. Specifically, I reveal the process of identifying the parameters of interest with discrete outcomes predicted from a continuous proxy variable. The key observation is that I do not need to invert the distribution function of the control outcome, allowing for a discrete outcome measure for point identification. Random assignment helps to justify this strategy. Particularly, randomization creates symmetry in the control outcome and proxy variable distributions across the units with and without the intervention. This symmetry allows us to directly compare the potential outcomes observed for units without an intervention versus units with an intervention, rather than going through the unifying structural functions. One can seek a proxy variable that is not necessarily the exact repeated measure because the structural function with the common latent rank in RCTs does not serve any essential role. This finding is particularly important in the context of RCTs because repeated measures may be unavailable. For example, an outcome such as timely graduation from a school cannot be repeated; conversely, the administrative records of test scores may be applicable as the graduation proxy. In this example, while the DiD does not apply, the rank imputation strategy may allow us to identify the assignment and treatment effects.

As the third contribution, I propose an estimator of the rank imputation strategy that is valid for any type of outcome measure and can flexibly incorporate covariates that have been compared with \citeauthor{AtheyImbens06}' (\citeyear{AtheyImbens06}) implementation. I demonstrate its desirable finite sample performance in clustered sampling as well as its robustness to nonlinear mean outcome functions that invalidate \citeauthor{AtheyImbens06}' (\citeyear{AtheyImbens06}) parametric implementation. Specifically, \cite{AtheyImbens06} proposed a nonparametric estimator with discrete covariates, but the implementation with continuous covariates took a restrictive parametric form. Similarly, \cite{MellySantanglo15} suggested a semiparametric estimator with covariates, in response to the general theory used by \cite{ChernozhukovFernandezValMelly13}; however, the strategy is less appealing for a CiC model without quantile functions of repeated outcome measures. Alternatively, I propose another semiparametric estimator based on distribution regressions. In Appendices A.1 and A.2, I discuss the weak convergence of the empirical process for the proposed counterfactual estimator and the validity of the bootstrap inferences. I also adopt \citeauthor{DaveziesDHaultfoeuilleGuyonvarch18}'s (\citeyear{DaveziesDHaultfoeuilleGuyonvarch18}) results of cluster-robust weak convergence to accommodate the experimental studies with cluster dependencies. I demonstrate the estimators' finite sample performance in simulations in Appendix E.1 and compared this with \citeauthor{AtheyImbens06}' (\citeyear{AtheyImbens06}) incorporation of covariates in Appendix E.2.

I apply the proposed estimator in the context of an experimental approach in development economics. In particular, I use the data from \cite{CreponDevotoDufloPariente2015}, who conducted a large-scale microcredit experiment in Morocco, to estimate the ITTTA, representing the intervention's combined effects. In addition, I separately estimate the ITTNA, which is the direct effect of the assignment for the nontreated, finding that the corresponding IV point estimate is 2.2 times larger than my preferred ITTTA estimate. Nevertheless, the estimated effects are not statistically different because of the large standard errors. The signs and magnitudes of the suggested ITTTA and ITTNA estimates may imply the possibility of a small, but positive direct assignment effect, while the treatment itself should have a large positive impact. As a result, I demonstrate that the small direct effect may be heavily biased in the ATT's IV point estimate, which should equal the ITTTA under the exclusion restriction. I also apply my strategy to \citeauthor{Gertler_Martinez_Rubio-Codina_2012}'s (\citeyear{Gertler_Martinez_Rubio-Codina_2012}) two-sided noncompliance experiment that studied the long-term impact of a cash transfer policy in Mexico. In particular, I evaluate their claim that the assignment's long-term effect is mainly caused by productive animal investments (See Appendix F for the description of the latter's application).

The rest of this paper is organized as follows. Section 2 introduces the notations and parameters of interest, Section 3 lists the formal identification assumptions illustrated in cases of one-sided noncompliance, and Section 4 demonstrates the estimation procedure with uniformly valid and cluster-robust bootstrap inferences. I apply the procedure to an experimental microcredit study in Section 5, discuss the framework in the case of two-sided noncompliance models in Section 6, and present my conclusions in Section 7.

\section{Parameters of interest}

Consider a standard model of potential outcomes. Let $T \in \{0,1\}$ be a binary randomized assignment and $Y$ be an observed outcome generated out of the potential outcomes. Assignment $T$ indexes the potential outcomes $Y(1)$ and $Y(0)$, such that $Y = T Y(1) + (1 - T)Y(0)$. Assignment $T$'s average effect
\[
 E[Y(1) - Y(0)] = E[Y(1)] - E[Y(0)],
\]
is the ITT. The ITT involves any assignment effects, including the associated treatment effect that assignment $T$ enables or enhances. Although this unconditional ITT represents assignment $T$'s average effect, a study's eventual goal often regards the associated treatment rather than the assignment.

Assume that the binary assignment of $T$ introduces another binary treatment $D \in \{0,1\}$. The units, which are the experimental subjects, choose $D$ endogenously after assignment $T$. If a unit is assigned to the control group ($T = 0$), then the treatment is not available ($D = 0$). The units with assignment $T = 1$ may not comply with the associated treatment $D$. Therefore, $D$ may be either $0$ or $1$ if the unit is in the treatment group $T = 1$.
\begin{definition}\label{def:onesided}
 A design $(Y,D,T)$ satisfies the one-sided noncompliance if 
 \[
  D = 
  \begin{cases}
   0 \mbox{ or } 1 & \mbox{ if } T = 1\\
   0  & \mbox{ if } T = 0.
  \end{cases}
 \]
\end{definition}
The following is an example of a one-sided noncompliance experiment.
\begin{example}
Microcredit experiments are an example of investigations in which researchers are aware of and interested in an assignment's direct impact. Specifically, \cite{CreponDevotoDufloPariente2015} studied the effect of introducing microcredit in Morocco's rural areas. They randomly assigned villages into treatment villages $(T = 1)$ and control villages $(T = 0)$, and the microcredit treatment $D$ followed one-sided noncompliance.
\end{example}
With treatment $D$, the observed outcome is 
\[
 Y = TY(1) + (1 - T)Y(0) = T(DY(1,1) + (1 - D)Y(1,0)) + (1 - T)Y(0)
\]
where $Y(1,1)$ is the outcome given the assignment $(T = 1)$ and taking up
the treatment $(D = 1)$; $Y(1,0)$ is the outcome given the assignment
$(T = 1)$, but not taking up the treatment $(D = 0)$; and $Y(0)$ is the outcome of
the control group, with $T = 0$ and $D = 0$. The leading parameter of interest is the
following ITT conditional on $T = 1$ and $D  = 1$,
\begin{equation}
 E[Y(1) - Y(0)|T = 1, D = 1] = E[Y(1,1) - Y(0)|T = 1, D = 1].\label{eq.ITTTA}
\end{equation}
I define the left-hand side of (\ref{eq.ITTTA}) as the ITTTA, where the equality to the right-hand side of (\ref{eq.ITTTA}) holds under Definition \ref{def:onesided}. Specifically, under Definition \ref{def:onesided},
the ITTTA is the combined effect of 
the ATT,
\[
 E[Y(1,1) - Y(1,0)|T = 1, D = 1],
\]
and assignment $T$'s effect on the treated, net of treatment $D$
\[
 E[Y(1,0) - Y(0)|T = 1, D = 1].
\]
The latter parameter is often called the direct effect of the treated in the principal stratification literature (\citealp{HiranoImbensRubinZhou2000}; \citealp{FrangakisRubin2002}; \citealp{FloresFloresLagunes2013}; \citealp{MealliPacini2013}) and mediation analysis literature (\citealp{Pearl14}). Examining if the ITTTA is significant confirms that experimental assignment $T$ operates through the expected channel of treatment $D$ since intention $T$ or treatment $D$ have impacts, at least for the intended units that comply with the intention. As an analogue parameter, I consider 
\begin{equation}
 E[Y(1) - Y(0)|T = 1, D = 0] = E[Y(1,0) - Y(0)|T = 1, D = 0], \label{eq.ITTNA}
\end{equation}
and I define the left-hand side of (\ref{eq.ITTNA}) as the ITTNA, where the equality holds under Definition \ref{def:onesided}. This ITTNA, the direct effect of the nontreated, has an important policy implication because it indicates if assignment $T$ has a direct impact on the outcome without considering the treatment.
The left-hand-side definitions of the ITTTA and ITTNA represent the same parameters of interest for both one-sided and two-sided noncompliance, while their underlying principal strata differ. Nevertheless, the above interpretation remains similar (see Section 6 for details). The distinction appears when we are interested in separating the net effect of $D$ from the ITTTA and ITTNA (see Section 3.3 for one-sided noncompliance and Section 6.3 for two-sided noncompliance).


In RCTs, $T$ is randomly assigned, but the treatment
 $D$ may be endogenous. For one-sided noncompliance, running a regression of $Y$ on $T$ and $D$ without covariates may produce a biased estimate in the coefficient of $D$ due to selection bias
\citep{AngristPischke09}. Namely,
\begin{align*}
 E&[Y|T = 1, D = 1] - E[Y|T = 1, D = 0] = E[Y(1,1)|T = 1,D = 1] - E[Y(1,0)|T = 1,D = 0]\\
 =& \underbrace{E[Y(1,1) - Y(1,0)|T = 1,D = 1]}_{\text{the ATT}} 
 + \underbrace{E[Y(1,0)|T=1,D = 1] - E[Y(1,0)|T=1,D = 0]}_{\text{selection bias}}.
\end{align*}

Therefore, researchers often employ random assignment $T$ as the
IV for taking up treatment $D$. Under three
conditions, (i) random assignment $T \indep Y(t)$, (ii) relevancy
$Cov(T,D) \neq 0$, and (iii) the exclusion restriction $E[Y(1,0) - Y(0)|T  = 1, D = 0] = 0$,
the IV estimator identifies the ITTTA, which equals to the ATT,
\begin{align*}
 \frac{Cov(Y,T)}{Cov(T,D)} =& \frac{E[Y(1) - Y(0)]}{P(D=1|T = 1)}\\
=& E[Y(1) - Y(0)|T = 1,D = 1] = E[Y(1,1) - Y(1,0)|T = 1, D = 1].
\end{align*}
where the last equality, based on the exclusion restriction (iii), may not hold. In particular,
assignment $T$ may affect the outcome for those who
do not take up treatment $D$.
\begin{example}
 \cite{CreponDevotoDufloPariente2015} listed several reasons why giving access to assignment $T$ may have an impact even for those who do not borrow $D$ from the microcredit institution:
\begin{quote}
 There are good reasons to believe that microcredit availability
 impacts not only on clients, but also on non-clients through a variety
 of channels: equilibrium effects via changes in wages or in
 competition, impacts on behavior of the mere possibility to borrow in
 the future, etc. (\citealp{CreponDevotoDufloPariente2015}, p. 124)
\end{quote}
 This possibility of the direct impact of availability not only complicates identifying the heterogeneous treatment effects but also emphasizes the importance of the direct impact for understanding the consequences of microcredit intervention.
\end{example}
If the exclusion restriction is violated, then the IV estimator also becomes a biased ITTTA estimator:
\begin{align*}
 \frac{Cov(Y,T)}{Cov(T,D)} =& E[Y(1) - Y(0)|T = 1,D = 1]\frac{P(D = 1|T = 1)}{P(D = 1|T = 1)} \\
&+ E[Y(1,0) - Y(0)|T=1,D = 0]\frac{P(D = 0|T = 1)}{P(D = 1|T = 1)}.
\end{align*}
The degree of bias depends on the
take-up rate $P(D = 1|T = 1)$. If $P(D = 1|T = 1) > 0.5$, then the
magnitude of the bias is less than $E[Y(1,0) - Y(0)|T  = 1, D = 0]$, but the IV point estimate may have been inflated in absolute value through divisions of the small probability $P(D = 1|T = 1)$ when the take-up rate is low.

Researchers frequently avoid using the IV estimator and report the ITT estimate because of the former issue of bias in the IV estimator. However, the ITT does not reveal the impact that appears because of the treatment in addition to the assignment (i.e., ITTTA), nor the impact appearing solely by the assignment (i.e., ITTNA).

\section{Identification}

In this section, I consider the identification assumptions for the ITTTA and ITTNA illustrated through one-sided noncompliance designs.

\subsection{Identification assumption}

If an individualistic assignment causes a direct effect by altering the behavior of units, then the potential outcomes, trivially defined as stable unit treatment value assumption (SUTVA) \citep{Rubin80}, can be justified. If there are spillover or equilibrium effects from community-level interventions, then SUTVA may be violated for the individual observation units. To accommodate the direct effects stemming from strategic interactions, as equilibrium behavior in taking up a treatment, I weaken SUTVA to create unique equilibrium potential outcomes that may depend on community $c$ to which unit $i$ belongs.
\begin{assumption}\label{ass:WDP}
 Let $c$ be a set of observations
 $\{1,\ldots,n_c\}$, with common assignment $T_c \in \{0,1\}$. For every $c, i \in c$, and $t \in \{0,1\}$, there are unique maps for a treatment indicator $D_{i;c}(t)$ and potential outcomes $Y_{i;c}(t,d)$ with
 the corresponding $d$ values in support of $D_{i;c}(t)$. 
\end{assumption}

%

Throughout this study, I assume that a random sample of units from the population consists of identical
units, representing either individuals or communities. For example, the ITTTA
\[
 E[Y_{i;c}(1,1) - Y_{i;c}(0)|T_c = 1, D_{i;c} = 1]
\]
is the mean comparison of potential outcomes $Y_{i;c}(1,1)$ and $Y_{i';c'}(0)$ from identical
communities $\{c,c'\}$ with $T_c = 1$ against $T_{c'} = 0$ for the same type of members $\{i,i'\}$ who would apply the treatment 
if their communities belonged to the treatment group $T_c = 1$, namely, $D_{i;c}(1) = D_{i';c'}(1) = 1$. Under Assumption \ref{ass:WDP}, the take-up and non-take-up subgroups
are stable and unique for a given draw of communities. 
Nonetheless, I do not observe these take-up or non-take-up members in these communities without assignment $T_c =
0$. Therefore, I need an additional identification restriction. Below,
I omit the index $\{i;c\}$ from the potential outcomes to simplify the notations.


Hereafter, I consider RCTs with the two consecutive surveys of
baseline and intervention. Let $Y_b$ denote a variable observed in the
baseline survey to ensure that $Y_b$ works as a proxy for the control
outcome of interest $Y(0)$. Let $W$ denote a vector of other pretreatment covariates
observed in the baseline survey. I assume that the random assignment of $T$ is
successful for the potential outcomes $Y(0)$ and $Y(1)$ and the baseline
variable $Y_b$.
\begin{assumption} \label{ass:CIZ}
 \[
  (Y(t),D(t)) \indep T | W, \forall t \in \{0,1\},
 \]
 and 
 \[
  (Y_b,D(t)) \indep T | W.
 \]
\end{assumption}

This randomization assumption is sufficient to identify the following ITT
\begin{lemma} \label{lmm:cATE}
 If Assumptions \ref{ass:WDP} and \ref{ass:CIZ} hold, then
 \[
  E[Y(1) - Y(0) | W = w] = E[Y | T = 1, W = w] - E[Y | T = 0,W = w]
 \]
 for every $w \in \mathcal{W}$. 
\end{lemma}
\begin{remark}
Assumption \ref{ass:CIZ} is stronger than necessary for the ITT in the independence of $Y_b$ from $T$. I label this assumption as random assignment because of the restriction on $Y_b$, notwithstanding its similarity to the conditional ignorability for observational data (see Remark 2 for the role of covariates $W$).
\end{remark}

My identification strategy is to exploit certain similarities in the proxy variable
$Y_b$ and control outcome $Y(0)$. However, these two random variables, $Y_b$
and $Y(0)$, may have different distribution functions. Thus, I
restrict the rankings of $Y_b$
and $Y(0)$ to ensure that they are similar. Throughout the study, $F_{Y|W}(\cdot|w)$ represents the conditional
distribution function of the random variable $Y$, and $Q_{Y|W}(\cdot|w)$ constitutes the conditional quantile function of $Y$.
\begin{definition} \label{def:latentRank}
 $W$ is a vector of baseline covariates and random variable $U_w \sim U[0,1]$, indexed by each $W = w$, is termed
 a (conditional) latent rank variable for random variable $Y$ if 
 \[
  Y = Q_{Y|W}(U_W|W)
 \]
 where $Q_{Y|W}(u|w) = \inf\{y: F_{Y|W}(y|w) \leq u\}$ and $u \in [0,1]$.
\end{definition}
The above definition is a key departure from other rank imputation strategies. Many approaches, including \citeauthor{AtheyImbens06}' (\citeyear{AtheyImbens06}) and \citeauthor{HuberSchelkerStrittmatter2019}'s (\citeyear{HuberSchelkerStrittmatter2019}), have assumed that a structural model shares the same scalar latent rank across potential outcomes. Such a structural assumption possibly imposes an obscure restriction on the nature of the random variable, especially when the outcomes are not continuously distributed. On the one hand, it is difficult to avoid such an additional structure function with a common latent variable across outcomes without randomization, as in Assumption \ref{ass:CIZ}. This is because there is no other way to convey the control outcome information from the no intervention units to the intervention units. On the other hand, the structure itself is more than necessary to identify the principal stratification effects when the groups are randomized. I do not assume that such a common latent rank exists across potential outcomes rather that latent rank variable $U_w$ exists for each potential outcome and is a definition as opposed to an assumption. Such a variable exists whether $Y$ is finitely supported or continuous. In fact, we can always construct such a conditional latent variable $U_w$.
Let, $F_{Y|W=w}(y-) \equiv \lim_{\tilde{y} \downarrow y}F_{Y|W=w}(\tilde{y})$, for the right-continuous conditional cumulative distribution function. Moreover,
 \[
  U_w = F_{Y|W}(Y-|w) + V \cdot (F_{Y|W}(Y|w) - F_{Y|W}(Y-|w))
 \]
 where $V \sim U[0,1]$ and $V \indep (Y,W)$, resulting in a quantile function $Q_{Y|W}(\cdot|w)$ that generates the random variable $Y$ (see \citeauthor{Ruschendorf2009}'s (\citeyear{Ruschendorf2009}) Proposition 2.1 for the presence of such a latent variable).


For the identification, I impose two restrictions on the relationship between
the proxy variable $Y_b$ and the control outcome $Y(0)$. First, I need to
determine the complete latent ranking of $Y_b$, $U_{b,w}$, over the entire support of
$[0,1]$. To achieve the requirement, I assume $Y_b$ has a
strictly increasing and continuous distribution
function.
\begin{assumption} \label{ass:UQ}
$\mathcal{W}$ is the support of $W$, and for every $w \in \mathcal{W}$, $F_{Y_b|W}(\cdot|w)$ is strictly increasing and continuous. 
\end{assumption}

This assumption is required to point identify the parameter of interest. Although a partial identification is possible with a finitely supported proxy variable $Y_b$, this is beyond the scope of this study, as my primary focus is point identification. 

Second, I assume that the latent rankings of $Y_b$ and $Y(0)$ are similarly associated with (endogenous) treatment $D$ given $T = 1$. The latent rankings of the same correlation with D is an example of when following rank similarity restriction holds. 
\begin{assumption} \label{ass:RS}
 $U_{b,w}$ is the latent ranking of $Y_b$ and $U_{0,w}$ is the
 latent ranking of $Y(0)$, as defined in Definition \ref{def:latentRank}. 
 \[
  U_{b,w} \sim U_{0,w} |W = w, T = 1,D = d
 \] 
 for each $d \in \{0,1\}, w \in \mathcal{W}_1$, where $\mathcal{W}_1$ is
 the support of $W$ conditional on $T = 1$.
\end{assumption}
\begin{remark}
The distributions of $Y_b$ and $Y(0)$ may differ arbitrarily, and $Y(1)$ and $Y(0)$ may be finitely supported for point identification. In Appendix C.6, I demonstrate that a binary $Y(0)$ may satisfy both Assumption \ref{ass:UQ} and \ref{ass:RS} with a continuous proxy $Y_b$.


The repeated measure of $Y(0)$ in the baseline, namely $Y_b(0)$, is not always available. The discreteness of $Y(0)$ further motivates the use of a nonrepeated $Y_b$ because DiD may predict $Y(0)$ outside the allowable range.
A $Y_b$ that solely determines $Y_b(0)$ given $W$ is a good proxy because the latent rank of $Y_b$ equals that of $Y_b(0)$, which can be identically distributed with the latent rank of $Y(0)$.
As another example of a structure to determine a good proxy, an underlying unobserved determinant $U$, in addition to identical shocks ($\tilde{U}_b, \tilde{U}_0$) called \textit{slippages}, can determine both $Y_b$ and $Y(0)$ (\citealp{HeckmanSmithClements97}). In both cases, conditioning on $W$ affecting $Y(0)$ and $Y_b$ may eliminate multidimensional relations that may violate Assumption \ref{ass:RS}. Appendix C.4 displays the conditions sufficient for Assumption \ref{ass:RS} using structural models with illustrations for the microcredit example (Example C.1).

Note that the rank similarity restriction is on the reduced-form latent ranks. Structural modeling on the latent ranks helps us in justifying the restriction, but the model is not necessary. Randomization helps us in proceeding with reduced-form latent ranks because $(Y(0),Y_b(0))$ are common across treatment and control groups. 
\end{remark}
In addition, I assume the following restriction for the support of the covariates $W$:
\begin{assumption} \label{ass:overlapSpp}
$\mathcal{W}_1$ is the support of $W$ conditional that $T = 1$, and
$\mathcal{W}_0$ is the support of $W$ conditional that $T = 0$. Assume that
 $\mathcal{W}_1 = \mathcal{W}_0 \equiv \mathcal{W}$.
\end{assumption} 
\begin{remark}
This assumption imposes a restriction on the population supports, and for the purpose of identification, one can justify Assumption \ref{ass:overlapSpp} with RCT data. 
\end{remark}

\subsection{Identification result}
Given the above restrictions, I present the main result. The target parameter is the distribution
$F_{Y(0)|T,D}(y|1,d)$, its quantile, and the mean computed from the
distribution.
\begin{theorem} \label{thm:IdRS}
 If Assumptions \ref{ass:WDP}, \ref{ass:CIZ}, \ref{ass:UQ}, \ref{ass:RS}, and
 \ref{ass:overlapSpp} hold, then
 \[
  F_{Y(0)|W,T,D}(y|w,1,d) =
 F_{Y_b|W,T,D}(Q_{Y_b|W,T}(\tau_{y,w}|w,1)|w,1,d), \tau_{y,w} \equiv F_{Y(0)|W,T}(y|w,0)
 \]
 for every $d \in \{0,1\}, w \in \mathcal{W}$ and $y$ in the support of $Y(0)|W=w,D=d$.
\end{theorem}
\begin{proof}
 See Appendix B.
\end{proof}

\begin{center}
  \begin{figure}[h!]
   \includegraphics[width=0.95\textwidth]{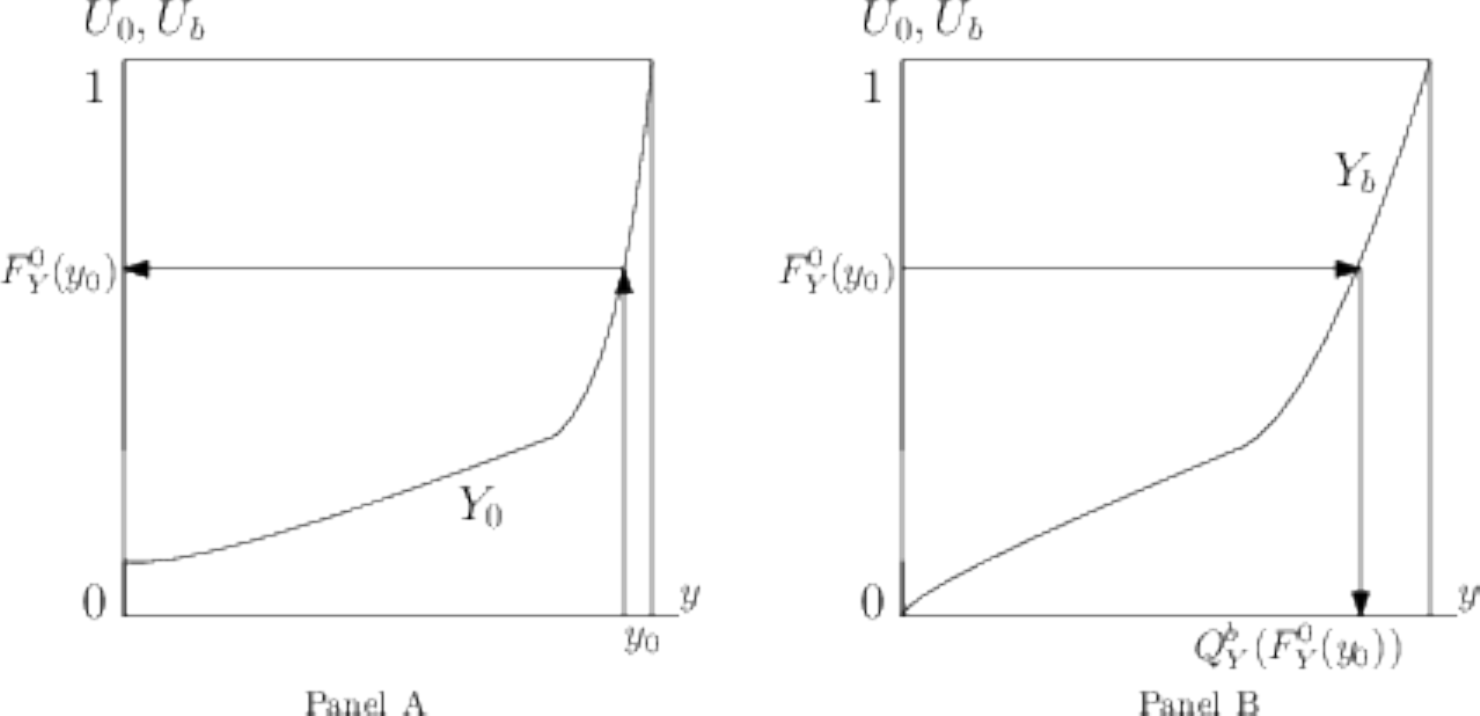}
 \caption{A graphical representation of the identification formula. In
   both panels, the $x$ axis shows the level of the outcome measures $Y(0)$ and
   $Y_b$, and the $y$ axis represents the level of the latent rank variables $U_0$ and
   $U_b$. Panel A offers the map from $y_0$ in the support of $Y(0)$
   into $F_{Y(0)}(y_0)$ in the support of $U_0$. Panel B
   presents the map from $F_{Y(0)}(y_0)$ in the support of $U_b$ into
   $Q_{Y_b}(F_{Y(0)}(y_0))$ in the support of $Y_b$.}
  \label{fig:Formula}
  \end{figure} 
\end{center}

Figure \ref{fig:Formula} summarizes the concept underlying the formula. For simplicity, we can consider the case without the covariates W. As Panel A in Figure \ref{fig:Formula} describes, the events $\{Y(0) \leq y\}$ and $\{U_0 \leq F_{Y(0)}(y_0)\}$ are equivalent almost surely for every $y_0$ in the support of $Y(0)$, $\mathcal{Y}_0$. Although $Y(0)$
may have a positive mass in the support of $Y(0)$, it does not hamper this step because we only need to evaluate $F_{Y(0)}(y_0)$ in the support of $Y(0)$. Rank
similarity equates the events of $\{U_0 \leq F_{Y(0)}(y_0)\}$ and $\{U_b
\leq F_{Y(0)}(y_0)\}$, with a conditional expectation of $\{T = 1, D =
d\}$. As Panel B in Figure \ref{fig:Formula} shows, we may find the quantile of
$Y_b$ in the support of the random variable $U_b$. Such a quantile
$Q_{Y_b}(u)$ can be found for every $u \in [0,1]$ from the continuity of
$Y_b$. Therefore, for every value of the latent rank $F_{Y(0)}(y_0)$, we
can indicate its quantile $Q_{Y_b}(F_{Y(0)}(y_0))$ in terms of $Y_b$. Then the events $\{U_b \leq
F_{Y(0)}(y_0)\}$ and $\{Y_b \leq Q_{Y_b}(F_{Y(0)}(y_0))\}$ have the same
probability conditional on $\{T = 1, D = d\}$. Thus, I equate the conditional expectation of $\{Y(0) \leq y\}$
with that of $\{Y_b \leq Q_{Y_b}(F_{Y(0)}(y_0))\}$. 

While standard CiC approaches (\citealp{AtheyImbens06}; \citealp{HuberSchelkerStrittmatter2019}) include a support condition for the latent rank variable, this condition is not presented in this study because I do not impose a structural function with a common latent rank variable across the functions. I define the latent ranks from the potential outcome as uniform random variables on $[0,1]$, and impose a direct pairwise relation in unobserved ranking variables $U_{0,w}$ and $U_{b,w}$. This direct relation is a distinction from the control function approaches (e.g., \citealp{ImbensNewey2009};
\citealp{DHaultfoeuilleFevrier2015}; \citealp{Torgovitsky2015};
\citealp{Ishihara2017})
 

\subsection{ATT identification with homogeneity restriction}

With the rank similarity binding the control outcome and proxy variable, I present the ITTTA and ITTNA identification process without, for example, imposing restrictions on $Y(1,0)$ and $Y(1,1)$. Although the ITTTA and ITTNA have important scientific and policy implications, I may be able to separate the ATT from the ITTTA if the following homogeneity assumption holds.
\[
 E[Y(1,0) - Y(0)|T = 1, D = 1] = E[Y(1,0) - Y(0)|T = 1, D = 0]. 
\]
%
The above assumption indicates that the direct effect $Y(1,0)-Y(0)$ has the same mean for the treated and nontreated. This homogeneity assumption is related to the exclusion restriction, $E[Y(1,0) - Y(0)|T = 1, D = 1] = 0$, and is substantially weaker. One can consider the exclusion restriction as a special case, as the homogeneous direct effect is 0. With the homogeneity, the ATT, $E[Y(1,1) - Y(1,0)|T = 1, D = 1]$, is simply the difference between the ITTTA and ITTNA, $E[Y(1,1) - Y(0)|T = 1, D=1]-E[Y(1,0) - Y(0)|T = 1, D = 0]$.
This homogeneity restriction may be justified when the assignment effect is universal across units. For example, if the assignment causes an equilibrium effect on the local prices of essentials, then the homogeneity assumption may hold. By contrast, if the assignment causes a spillover effect that affects the treated and nontreated differently, then the assumption is implausible. Although this restriction is strong, I have two justifications for the above quantity. First, it is challenging to separate the ATT from the ITTTA unless scalar unobservable restrictions are imposed on both $Y(1,0)$ and $Y(1,1)$. Second, the above quantity may still operate as a back-of-the-envelope calculation of the magnitude of the ATT. For example, as the spillover of microcredit access may represent a transfer from the treated to the nontreated, the direct effect for the treated may be lower than that of the nontreated. In that case, the ATT value may be underestimated with the back-of-envelope calculation (see Section 5 for its application to the microcredit experiment and Section 6 for the complexities of this concept for two-sided cases).




\section{Estimation}

As demonstrated in Section 3 and Appendix C.4, it is desirable to condition the distributions on pretreatment covariates $W$ to justify the conditional rank similarity. Dimensionality becomes a serious issue when we nonparametrically estimate conditional distribution functions for each subsample $\{T=1,D=d,W=w\}$, with W containing continuous variables. 
\cite{MellySantanglo15} extended \citeauthor{AtheyImbens06}' (\citeyear{AtheyImbens06}) estimator to incorporate covariates into a semiparametric quantile regression model. However, semiparametric quantile regressions remain flexible only when they involve a smooth conditional density (see \citealt{ChernozhukovFernandezValMelly13}).
As my identification analysis now accepts $Y(0)$ as a binary or discrete random variable, the quantile regression approach may be undesirable. Therefore, I consider a distribution regression-based approach and incorporate \citeauthor{DaveziesDHaultfoeuilleGuyonvarch18}'s (\citeyear{DaveziesDHaultfoeuilleGuyonvarch18}) recent development of cluster-robust inferences for the empirical process. Finally, cluster dependency is inevitable in my leading examples and my inference strategy closely follows that of \cite{ChernozhukovFernandezValMelly13}.


I begin with a random sample of observations. Consider an estimation of the semiparametric conditional distribution functions. To estimate the parameter of interest, I first estimate the following distribution functions:
\[
 F_{Y_b|T,W}(\cdot|1,w),F_{Y_b|T,W}(\cdot|0,w),\{F_{Y_b|T,D,W}(\cdot|1,d,w)\}_{d \in
\{0,1\}}.
\]
Hereafter, $W$ is a transformed vector of the original pretreatment covariates, such as polynomials or
B-splines, and $K$ is the set of indexes for the subgroups, such as $\{T = 1,
D = d\}$ or $\{T = 1\}$. Following \cite{ForesiPeracchi1995} and \cite{ChernozhukovFernandezValMelly13},
I estimate the conditional distribution functions for $Y$ conditional on a
subgroup $k$ out of $K \equiv \{\{T = 1\},\{T = 0\},\{T = 1,D = 1\}, \{T = 1, D = 0\}\}$ and $W$ as
\[
 \hat{F}_{Y|W,K}(y|w,k) = \Lambda (w'\hat{\beta}^{k}(y))
\]
for some known link function $\Lambda(\cdot)$. I estimate $\hat{\beta}(y)$ as
\begin{align*}
 \hat{\beta}^{k}(y) = \arg \max_{b \in \mathbb{R}^{dW}} \sum_{i = 1}^{n}&\left\{
 \left[I\{Y_i \leq y\}\log[\Lambda(W_i'b)] \right]\right.\\
 &+ \left.\left[I\{Y_i > y\}\log[1 - \Lambda(W_i'b)] \right]\right\}I\{K_i
 = k\}
\end{align*}
for each $y \in \mathcal{Y}^k$ where $\mathcal{Y}^k$ is the support of $Y$
conditional on the subgroup $k$ and $dW$ is a dimension of $W$. 
For the main application in Section 5, I use the logit link function. Appendix D.2 shows the robustness to other link function choices such as probit link and complementary log-log link. 

Once these estimators are obtained, the conditional counterfactual
distribution is
\[
 \hat{F}_{Y(0)|W,T,D}(y|w,1,d) =
 \hat{F}_{Y_b|W,T,D}(\hat{Q}_{Y_b|W,T}(\hat{F}_{Y(0)|W,T}(y|w,0)|w,1)|w,1,d)
\]
where
\[
 \hat{Q}_{Y_b|W,T}(\tau|w,1) = \inf\left\{y \in \mathcal{Y}^{1,w}_b : \hat{F}_{Y_b|W,T}(y|w,1)
 \geq \tau \right\},
\]
in which $\mathcal{Y}_b^{1,w}$ is the support of $Y_b$ conditional on
$\{T = 1, W = w\}$ and $\tau \in [0,1]$ is a corresponding percentile. Therefore, the unconditional distribution can be obtained by
\[
 \hat{F}_{Y(0)|T,D}(y|1,d) = n_{1,d}^{-1}\sum_{i=1}^{n} \hat{F}_{Y(0)|W,T,D}(y|W_i,1,d)I\{T_i=1,D_i=d\}
\]
where $n_{1,d} \equiv \sum_{i} I\{T_i = 1,D_i=d\}$. The mean effect of
interest is obtained as follows
\[
 \hat{\mu}_d = \frac{1}{n_{1,d}} \sum_{i:T_i=1,D_i=d} Y_i  -
 \int_{\mathcal{Y}^0} y d\hat{F}_{Y(0)|T,D}(y|1,d)
\]
and the quantile difference is obtained by inverting the distribution functions
\[
 \hat{Q}_{Y(1)|T,D}(\tau|1,d) - \hat{Q}_{Y(0)|T,D}(\tau|1,d), \forall \tau \in (0,1).
\]

For the random sample of clusters, rather than individual
observations, I estimate
\[
 \hat{F}^C_{Y|W,K}(y|w,k) = \Lambda (w'\hat{\beta}^{C,k}(y))
\]
with
\begin{align*}
 \hat{\beta}^{C,k}(y) = \arg \max_{b \in \mathbb{R}^{dW}} &\sum_{c =
 1}^{\bar{C}}\sum_{i = 1}^{N_c}\left\{ \right.
 \left[I\{Y_{i;c} \leq y\}\log[\Lambda(W_{i;c}'b)] \right]\\
 &+ \left. \left[I\{Y_{i;c} > y\}\log[1 - \Lambda(W_{i;c}'b)] \right]\right\}I\{K_{i;c}
 = k\}
\end{align*}
for each $y \in \mathcal{Y}^k$, where $\bar{C}$ is the number of 
clusters and $N_c$ represents the size of each cluster $c$.
The difference appears in the bootstrap procedure and a few modifications under the assumptions of the data-generating process. In Appendices A.1 and A.2, I show that the proposed estimators are uniformly asymptotically normal. As a result, (clustered) exchangeable bootstrap inferences are shown to be valid for generating uniform confidence intervals (see also Appendices E.1 and E.2 for this estimator's finite sample properties).

\section{Application}

\subsection{Microcredit experiment in Morocco: background}

\cite{CreponDevotoDufloPariente2015} conducted a one-sided noncompliance experiment in the rural areas of Morocco, choosing target areas where the participant villagers had not experienced microcredit services before the experiment. This location choice was a novel feature of the study, and the authors estimated the relative effects of new access and no access, rather than the effect of expanding microcredit. Here, as introduced in Section 3, $T$ denotes the binary assignment (i.e., access to microcredit services) and D is the binary treatment of taking up microcredit. By construction, people in control villages, $T = 0$, automatically have $D = 0$. The administrative observation of the take-up decision validates the successful implementation of this procedure. Overall, the take-up rate is relatively low. As reported in column (1) of Table 2 in \cite{CreponDevotoDufloPariente2015}, only $16.8$\% of the treated units in the sample took a loan from the microcredit company. This number reduces to $14.3$\% with non-zero business outputs in the baseline survey.

For the present analysis, the most important feature is that \cite{CreponDevotoDufloPariente2015} conducted a detailed baseline survey before the experimental intervention. The dataset comprised the production output market values before and after the experiment. Although this baseline survey feature is common in field experiments, only a few studies allow the application of my strategy, as they do not report take-up behaviors in the baseline sample. Based on the aforementioned information, $Y_b$ denotes the baseline sales value output and $Y = T Y(1) + (1 - T)Y(0) = T(DY(1,1) + (1 - D)Y(1,0)) + (1 - T)Y(0)$ indicates the endline sales value output.

In the experiment, a local microfinance institution called Al Amana entered randomly selected villages. After Al Amana opened new branches at the beginning of the study, the authors conducted a baseline survey containing all the outcome measures of interest as the terminal outcome measures. Upon completing the baseline survey, a randomization process separated $162$ villages into $81$ pairs with similar observed characteristics; one of each pair was randomly assigned as the treatment village and the other as the control village. Al Amana agents visited the treatment villages and promoted participation in microcredit, while the control villages had no access. At the time of intervention, all the newly opened branches offered fully functional services. For the control group, we can confirm that this one-sided noncompliance feature, in the administrative report of the proportion of Al Amana clients, is zero. The main sample in \cite{CreponDevotoDufloPariente2015} is only a subsample (not the entire sample) of units that had a high borrowing probability. Although the study collected baseline survey observations for $4,465$ households, from $162$ villages, they also added the endline observations, totaling $4,934$ observations, as their analysis did not necessarily require a baseline survey structure. I use their extended sample in column (1) of Tables \ref{tbl:originalITT} and \ref{tbl:originalIV} to replicate their original estimates, but I limit the sample to the baseline survey units in the remainder of the study.

\begin{table}[tbh] 
 \caption{Intention-to-treat effects on business outcomes} 

 \label{tbl:originalITT}
 \centering
 \begin{tabular}{lcccc}
 \hline
 & (1) & (2) & (3) & (4)\\
Model & OLS & OLS & OLS & OLS\\ \cline{2-5}
Outcome & \multicolumn{2}{c}{level output} & \multicolumn{2}{c}{log output}\\ \hline
  Assignment ($T$) & $6,061^{***}$ & $5,889^*$ & $0.32^{***}$ & $0.37$\\
            & ($2,167$) & ($3,039$) & ($0.12$) & ($0.25$)\\
  \hline
  Self-employed in baseline &      & Y & Y & Y\\
  Strata dummies & Y & Y & Y & \\
  Obs       & $4,934$ & $2,453$ & $2,453$& $2,453$\\
  control mean & $30,450$ & $33,554$ & $8.71$& $8.71$\\
  \hline
 \end{tabular}
 \vspace{0.3cm}
\begin{minipage}{325pt}
{\flushleft
 \fontsize{9pt}{9pt}\selectfont \smallskip NOTE: Standard
 errors reported in parenthesis are clustered in
 village levels. *, ** and *** indicate statistical significance of
 10\%, 5\%, 1\% sizes respectively. Units in levels are Moroccan Dirham,
 1MAD $\approx $ $0.11$ USD. Self-employed indicates that the estimates
 are for the subgroup of business owners at the baseline. 
 } 
\end{minipage}
\end{table}

Table \ref{tbl:originalITT} shows the ITT estimates using regression analysis to control for the covariates, including number of household members, number of adults, the household head's age, indicator variables for animal husbandry, other nonagricultural activities, outstanding loans over the past 12 months, household spouses as the survey respondents, and other household members acting as the survey respondents. These linear regression analyses are also conditioned for strata dummies (paired villages) except for column (4). 

As column (1) shows, there is a positive and significant effect of microcredit access on production output sales. As my method requires a continuous $Y_b$, I restrict the study sample to individuals who had positive sales values in the baseline survey. As emphasized earlier, $Y(0)$ and $Y(1)$ do not need to be continuous. Conditional on $Y_b > 0$, some observations with $Y = 0$ represent exit behaviors during the study period, and the mass at $Y = 0$ does not hamper the analysis.

This sample selection may change the interpretation of the effect but does not generate any bias. Columns (2) and (3) show the same estimates, except for self-employed individuals at the baseline. This procedure reduces the original sample size to approximately half. Column (2) shows the effect on the output level and column (3) shows the effect on the output log. As the output values contain zeroes, I apply the inverse hyperbolic sine transformation, $log(x + \sqrt{x^2+ 1})$, instead of the natural log. Therefore, the estimates interpret the approximated semi-elasticities for small effects. Overall, we need to convert the estimates using the hyperbolic sine formula to interpret the large coefficients, whereas the standard exponential approximation works well for this application to evaluate large means (for a detailed discussion, see \citealp{BellemareWichman2018}). Both effects are positive and precisely measured, and the remainder of the arguments are all based on the log output, as the revenue distribution is heavily right skewed with a few outliers. It appears that the issue of outliers in the output values does not matter extraordinarily. Appendix D.3 presents additional results using the output level with or without trimming the extreme output values, which are consistent with the log output findings in the main analysis.

Table \ref{tbl:originalIV} shows the original and additional IV estimates that are valid only if the direct effect for the treated is zero. Under the conventional assumption, the estimated ATT may overestimate the effect for the treated, implying that policymakers may be overly encouraged to promote microcredit services. Thus, I compare my preferred estimate for the ITTTA with these two-stage least squares (2SLS) ATT/ITTTA estimates. The 2SLS estimates in columns (2) and (3) include paired village dummies, following \citeauthor{CreponDevotoDufloPariente2015}'s (\citeyear{CreponDevotoDufloPariente2015}) original specification. For my later estimates, I do not include these paired village dummies to prevent the incidental parameters problem for nonlinear estimators. Nevertheless, the 2SLS estimate with dummies in column (4) has a similar magnitude to the estimate without dummies in column (3), indicating that strata dummies improve their precision.
\begin{table}[tbh] 
 \caption{ATT under assumption of no direct effect} 

 \label{tbl:originalIV}
 \centering
 \begin{tabular}{lcccc}
\hline
 & (1) & (2) & (3) & (4)\\
Model & 2SLS & 2SLS & 2SLS & 2SLS \\ \cline{2-5}
Outcome & \multicolumn{2}{c}{level output} & \multicolumn{2}{c}{log output}\\ \hline
  Treatment ($D$) & $36,253^{***}$ & $42,026^*$ & $2.29^{***}$ & $2.57$\\
            & ($12,494$) & ($21,525$) & ($0.85$) & ($1.71$)\\
  \hline
  Self-employed in baseline &      & Y & Y & Y\\
  Strata dummies & Y & Y & Y & \\
  Obs       & $4,934$ & $2,453$ & $2,453$  & $2,453$\\
  control mean & $30,450$ & $33,554$ & $8.71$ & $8.71$\\
  \hline
 \end{tabular}
 \vspace{0.3cm}
\begin{minipage}{325pt}
{\flushleft
 \fontsize{9pt}{9pt}\selectfont \smallskip NOTE: Standard
 errors reported in parenthesis are clustered in
 village levels. *, ** and *** indicate statistical significance of
 10\%, 5\%, 1\% sizes respecively. Self-employed indicates that the estimates
 are for the subgroup of business owners at the baseline. 2SLS stands for the two-stage least squares method.
 } 
\end{minipage}
\end{table}

\subsection{ITTTA and ITTNA estimations}

With the baseline outcome $Y_b$ as the proxy for the control outcome $Y(0)$,
I can directly identify the counterfactual distribution of $Y(0)$ conditional on the endogenous subgroup $\{T = 1, D = d\}$ for each $d \in \{0,1\}$ (see Section 3 for details of this procedure). 

Outcomes $Y_b$ and $Y(0)$ are the production output sales values stemming from small business activities. Based on the random assignment of credit access $T$ and the fact that microcredit was not available during the baseline period or for control villages, these two outcomes, $Y_b$ and $Y(0)$, should be similar, except for the random shocks occurring over the study's two-year period. Appendix C.4 offers detailed arguments to justify the rank similarity based on the so-called \textit{slippages} argument (\citealt{HeckmanSmithClements97}). In this study, although the rank similarity assumption does not have any testable restrictions per se, I conducted a diagnostic test supporting the rank similarity assumption in this application (see Appendix D.4 for results and other details).


Table \ref{tbl:mainResults} shows the estimates of the subgroup effects
unconditional on $W$ achieved by integrating $W$ out.
Column (1) represents the ITTNA, the direct effect for the nontreated $D = 0$ with assignment $T = 1$. Column (2) addresses the ITTTA, as the combined effect of taking up treatment $D = 1$, with takers gaining access to assignment $T = 1$. In column (3), I offer a back-of-the-envelope calculation of the treatment's net effect, or the ATT, as discussed in Section 3.3.
\begin{table}[tbh]
 \centering
  \caption{Estimates with the baseline proxy $Y_b$} 

 \label{tbl:mainResults}
 \begin{tabular}{lccc}
 \hline
 & (1) & (2) & (3)\\
Model & RS & RS & RS \\ 
Parameter & ITTNA & ITTTA & ITTTA - ITTNA \\ \cline{2-4}
Outcome & \multicolumn{3}{c}{log output}\\
  \hline
  Assignment ($T$) & $0.28$ & $1.03^{**}$ & $0.74^*$ \\
  by subgroups of ($D$)    & ($0.28$) & ($0.45$) & ($0.44$) \\
  \hline
  Self-employed in basline & Y & Y & Y\\
  Obs       & $2,453$ & $2,453$ & $2,453$ \\ 
  \hline
 \end{tabular}
 \vspace{0.3cm}
\begin{minipage}{325pt}
{\flushleft
 \fontsize{9pt}{9pt}\selectfont \smallskip NOTE: Standard
 errors reported in parenthesis are generated from 300 bootstrap draws clustered in
 village levels for (1)-(3). *,**,*** indicates statistical significance of
 10\%,5\% and 1\% sizes respectively. Logit link is used for (1) and
 (2). RS stands for the rank similarity estimator proposed in this paper.
 } 
\end{minipage}
\end{table} 

As shown in column (2), the ITTTA is strongly positive and significant, but its magnitude is less than $45$\% of the 2SLS estimate, which should have been similar to the ITTTA if the direct effect for the treated was zero. A cautionary note when interpreting the difference between the two point estimates for the treated is that it is not statistically significant. As the p-value of the test for this difference is $0.182$, it is not possible to conclude that the proposed estimate implies that the 2SLS estimator is biased. Indeed, this insignificance appears in column (1), as the ITTNA, representing the estimated direct effect for the nontreated, is insignificant. 
If the direct effects are homogeneous, then the magnitude and insignificance of the ITTNA may correspond to the magnitude and insignificance of the difference between the ITTTA and IV. By combining the ITTTA and ITTNA signs and magnitudes, $T$ may still have a positive direct but imprecisely estimated effect. Nevertheless, we cannot distinguish imprecisely estimated small effect against no effect.

Column (3)'s ATT back-of-the-envelope calculation suggests a sufficiently large effect from the treatment itself. Although the validity of this estimate depends on the homogeneity of the direct effect, the direct effect may be smaller for the treated than the nontreated, as discussed in Section 3.3. Therefore, it is assuring to see the positive and marginally significant effect confirming the ATT's positivity. In Appendix D.1, I display the quantile differences of the direct effects for the nontreated, $D = 0$, (Figure 1) and the quantile difference of the combined effects for the treated, $D = 1$, (Figure 2), along with uniform $95$\% confidence intervals. The results indicate the relatively monotonous effects over the range of quantile values for the combined effect of the treated. Specifically, the quantile differences are conditional on those who had businesses at the baseline. Therefore, this result does not contradict previous findings that business revenue effects appear in the upper tail of the distributions (\citealp{AngelucciKarlanZinman15}).

Although Table \ref{tbl:mainResults} findings may alter policymakers' decisions, I must clarify two limitations of this analysis. First, the focus of this study was limited to the subsample of business owners before the experiment. Second, the conventional IV estimate and my preferred estimate differ, but not statistically. Therefore, my empirical results should be considered as additional evidence for policymakers. I encourage future researchers to collect a detailed baseline survey to evaluate the intervention of interest and further understand the nature of the treatment effect through the ITTTA and ITTNA.



\section{Two-sided noncompliance}

Thus far, I have focused on one-sided noncompliance cases because two-sided noncompliance designs are complicated by the dynamic nature of taking up the treatment. This complication further distinguishes my analysis from the previous literature, including \cite{HuberSchelkerStrittmatter2019}, by selecting plausible kinds of latent rank restrictions separately for each design. In two-sided noncompliance cases, the practical complication involves justifications for the rank similarity assumption for which the rationale is that the underlying latent rank of the control outcome $Y(0)$ and proxy variable $Y_b$ would be similar conditional on the endogenous subgroups of $D$ and $T$. In one-sided noncompliance cases, the similarity of $Y(0)$ and $Y_b$ is a natural concept because both variables are determined when neither $T$ nor $D$ is available. For two-sided noncompliance designs, $Y(0)$ is no longer a purely comparable outcome with $Y_b$. As the control units may have taken up treatment $D$ without assignment $T$, $Y(0)$ represents the outcomes nontreated $D = 0$ and treated $D = 1$ under $T = 0$. Below, I articulate my argument for two-sided noncompliance designs in two cases.

If $D(1)$ and $D(0)$ denote the potential choice under assignment $T$, then $D = TD(1)+(1-T)D(0)$. A two-sided noncompliance experiment is a design $(Y,D,T)$ generated from $D(t) \in \{0,1\}$ for both $t \in \{0,1\}$. The observed outcome takes the following form: 
\begin{equation}
 Y = T (D(1) Y(1,1) + (1 - D(1))Y(1,0)) + (1 - T)(D(0) Y(0,1) + (1 -
  D(0))Y(0,0)) \label{eq:twoSided}
\end{equation}
where $Y(t,d)$ represents the outcome with assignment $T = t$ and treatment $D = d$.

The main parameters stay the same as the ITTTA and ITTNA. Under the monotonicity of $D(1) \geq D(0)$ almost surely, they have similar interpretations. The ITTTA is a principal stratification effect of $E[Y(1,1) - Y(0,D(0))|T = 1, D(1) = 1]$ for the union of basic principal stratum $\{D(1)=1,D(0)=0\} \cup \{D(1)=1,D(0)=1\}$. The ITTTA is also the sum of the effect from the intended channel for the complier, $E[(1 - D(0))(Y(1,1) - Y(1,0))|T = 1, D(1) = 1]$, and the assignment's direct effect for the treated holding the same treatment status in the control, $E[Y(1,D(0)) - Y(0,D(0))|T = 1, D(1) = 1]$. Therefore, verifying the ITTTA can still determine whether $T$ has an impact, at least for those who comply with intention $T$. The ITTNA is the same direct effect $E[Y(1,0)-Y(0,0)|T=1,D(1)=0] = E[Y(1,0) - Y(0,0)|D(1) = D(0) = 0]$ by monotonicity. Nevertheless, the identification rationale may differ.


\subsection{Treatment available before the randomized assignment}



Consider a two-sided noncompliance experiment where treatment $D$ may be available before the assignment of $T$. If the same treatment $D$ had been available in a previous period, but assignment $T$ encouraged taking up $D$, then the rank similarity strategy for one-sided noncompliance cases may still work out. Nevertheless, the justification for Assumption \ref{ass:RS} should be based on $Y(0,1)$ and $Y(0,0)$, instead of $Y(0)$. The key notion is the stability of taking up treatment $D$ within the units without the assignment. While $D$ was available in the baseline, no factors should affect take-up behavior $D$ in the control group $T = 0$. Therefore, the $D$ distribution should not change in the control group over time. If $D_b$ is the treatment in the baseline and $Y_{b}(0)$ and $Y_{b}(1)$ are the baseline proxies with and without $D_b$, then $Y_b \equiv Y_{b}(1)D_b + Y_{b}(0)(1 - D_b)$. With these notations, I propose the following modified rank similarity assumption.
\begin{assumption} \label{ass:RSTS}
 Suppose that Assumption \ref{ass:overlapSpp} holds for the simplicity.
 Then, $U_{b,d,w}$ is the latent ranking of $Y_{b}(d)$ for $d \in \{0,1\}$, and $U_{0,d,w}$ is the
 latent ranking of $Y(0,d)$ for $d \in \{0,1\}$, as defined in Definition \ref{def:latentRank}. This allows us to assume that
 \[
  P(U_{b,d,w} \leq u|W=w, D(1)= d', D_b = d) = P(U_{0,d,w} \leq u|W=w, D(1)= d', D(0) = d)
 \] 
 for each $(d,d') \in \{0,1\}^2, w \in \mathcal{W}$.
\end{assumption}
This modified rank similarity imposes that the latent rankings do not change in $Y(0,d)$ and $Y_{b}(d)$ for those who maintained the same treatment status $D(0) = D_b = d$ without assignment $T = 0$, and who would behave similarly in $D(1)$. With this modified rank similarity, the ITTTA and ITTNA are identified when the compliance rates are \textit{stable} over time without the assignment (see \citealt{deChaisemartinDHaultfoeuille2017}, Assumption 2).
\begin{corollary} \label{cor:RSTS}
  Suppose that Assumptions \ref{ass:WDP}, \ref{ass:CIZ}, \ref{ass:UQ},
 \ref{ass:overlapSpp} and \ref{ass:RSTS} hold. Furthermore, assume that
 \[
  P(D(0) = 1|D(1) = d, W = w) = P(D_b = 1|D(1) = d, W=w)
 \]
 for each $d \in \{0,1\}$ and $w \in \mathcal{W}$, and $D(1) \geq D(0)$ almost surely. Then, we have
 \[
  F_{Y(0)|W,T,D}(y|w,1,d) = F_{Y_b|W,T,D}(Q_{Y_b|W,T}(\tau_{y,w}|w,1)|w,1,d), \tau_{y,w} \equiv F_{Y(0)|W,T}(y|w,0)
 \]
 for every $d \in \{0,1\}, w \in \mathcal{W}$ and $y$ in the support of $Y(0)|W=w,D=d$.
\end{corollary}
\begin{proof}
 From Lemma C.1 in Appendix C.1, Assumption \ref{ass:RS} holds for $Y(0) = Y(0,1)D(0)  + Y(0,0)(1 - D(0))$. The proof concludes by satisfying all the conditions for Theorem \ref{thm:IdRS}.
\end{proof}

\subsection{Two-sided noncompliance, but not before the experiment}

Next, I consider treatment $D$ introduced after assignment $T$. In this case, it is difficult to justify the rank similarity in an analogous manner. While $Y(0)$ now contains the outcome of taking up treatment $D$, $Y_b$ cannot have this component in the baseline. Thus, I consider a modification to the rank similarity to deal with this difficulty. The weakened rank similarity assumption imposes a restriction only for the nontreated, solely restricting on the latent ranks of $Y(0,0)$ and $Y_b$ given $D = 0$ and $T = 0$.
\begin{assumption} \label{ass:RSTS2}
  Suppose that Assumptions \ref{ass:WDP}, \ref{ass:CIZ} and \ref{ass:overlapSpp}, $D(1) \geq D(0)$ almost surely, and for every $w \in
 \mathcal{W}$, $(U_{b,0,w},U_{0,0,w})$ represents the latent rankings of $Y_b$ and $Y(0,0)$ conditional on $T = 0, D
 = 0, W = w$, then we can assume that
 \[
  U_{b,0,w} \sim U_{0,0,w}|D(1) = 0, D(0) = 0, W = w.
 \]
\end{assumption} 
Given this weaker restriction, I show the ITTNA identification process, but not for the ITTTA.
\begin{theorem} \label{thm:twosidedID} 
 Suppose that Assumptions \ref{ass:RSTS2} and \ref{ass:UQ} hold, with the distribution of $Y_b$ being conditional on $\{T = 0, D = 0, W = w\}$ for every $w \in \mathcal{W}$. Then,
 \begin{align*}
 &F_{Y(0)|D,T}(y|0,1)\\ 
 &= \int
 F_{Y_b|D,T,W}(Q_{Y_b|D,T,W}(F_{Y(0,0)|D ,T,W}(y|0,0,w)|0,0,w)|0,1,w)dF_{W|D,T}(w|0,1)  
 \end{align*}
and therefore, we identify $E[Y(1) - Y(0)|D = 0, T = 1] = E[Y(1,0) - Y(0,0)|D(1)=0].$
\end{theorem}
\begin{proof} 
 See Appendix B.
\end{proof}
%

The ITTTA may still be identified from the total law of expectation using the ITT estimate. In other words, $ITTTA = ITT - ITTNA \times \frac{P(D=0|T = 1)}{P(D=1|T = 1)}$. Nevertheless, it is limited in its precision when it must be divided by the small probability of $P(D = 1|T = 1)$ if the take-up rate is low. In addition, the ITTNA can be the only parameter of interest in some two-sided noncompliance designs. In Appendix C.4, I provide an example of the selective attrition problem, which is illustrated with \cite{vandenBergVikstrom2014}, to show the usefulness of Theorem \ref{thm:twosidedID} on its own.

\subsection{Identifying other parameters}
The limited identifiability in Theorem \ref{thm:twosidedID} derives from the restriction's weakness, which is typical of a two-sided noncompliance design. Further, the proposed rank similarity approach itself cannot identify the other parameters studied in \cite{HuberSchelkerStrittmatter2019}, such as the net (indirect) effects of $D$ for the compliers or other principal strata. The critical difference lies in their restrictions on the model. As with other CiC models, \cite{HuberSchelkerStrittmatter2019} assumed that the potential outcomes $Y(t,d)$ share a common latent rank variable across $(t,d) \in \{0,1\}^2$. Given that the same single unobserved scalar governs all the outcomes with and without assignment $T$ within a strong model, the restrictions on the common latent variable allow them to identify all the mediation parameters. \cite{HuberSchelkerStrittmatter2019} successfully found the mechanical latent rank conditions needed for each mediating effect, but the rationales for each latent rank restriction is not provided. Specifically, for operations involving the treated outcomes $Y(1)$, it is challenging to justify the associated latent rank restrictions. Further, it is difficult to weaken \cite{HuberSchelkerStrittmatter2019} restrictions because the latent rank restrictions are not separated from the model.

In this study, I do not impose such a model, but the potential outcome structure. With only the potential outcome structures, I explicitly declare which latent rank restrictions could be plausible under what kinds of designs, thereby offering a conservative approach to identifying the key principal stratification parameters. Nevertheless, I consider the net effects of treatment $D$, which is the ATT, with a homogeneity assumption rather than the latent rank restrictions for $Y(1,0)$ or $Y(1,1)$, as shown in Section 3.3. A similar procedure is possible in two-sided noncompliance designs. In Appendix C.2, I describe the identification process of the following \textit{net} effect of $D$ for the complier, similar to the local average treatment effect (LATE):
\begin{align}
 E[Y(0,1) - Y(0,0)|D(1) > D(0)] = (ITTTA - ITTNA)\frac{E[D(1)]}{E[D(1)] - E[D(0)]}\label{eq:LATE}
\end{align}
under a homogeneity assumption, $E[Y(1,1)-Y(0,1)|D(1) = 1] = E[Y(1,0) - Y(0,0)|D(1) = 0]$, that is stronger than that discussed in Section 3.3. The formula follows because $ITTNA = E[Y(1,0) - Y(0,D(0))|D(1) = 0] = E[Y(1,0) - Y(0,0)|D(1) = 0]$ when $D(1) \geq D(0)$ almost surely. As in Section 3.3, the above net effect may still work as a back-of-the-envelope calculation when the homogeneity assumption is violated. (see Appendix F for the empirical application of two-sided noncompliance and homogeneity assumption justifications).




\section{Conclusions}
\label{sec:conc}

This study presents a method that does not require any additional instruments, treatment exogeneity, or specific experimental designs to identify treatment effect heterogeneity across endogenous decisions in randomized experiments. In contrast to the existing literature, I used a variable from a baseline survey to proxy for the control outcomes. I then offered a procedure to directly identify the ITTs conditional on endogenous strata using proxy variables from the baseline survey. My chosen method produced three distinct novelties. First, exploiting the baseline survey—a typical feature of experimental studies—the relatively conservative natural experimental approach in RCTs offers fewer parameters, which are identified through substantially fewer model structures and weaker restrictions, than in the recent work of \cite{HuberSchelkerStrittmatter2019}. Second, I only required the continuous variable to act as a proxy variable, not as the outcome of interest. Unlike in standard CiC models, I expanded the considerations of the proxy variable, not necessarily representing a repeated outcome measure, and this flexibility allowed for identification with discrete outcomes. This feature is critical, as the exact repeated outcome measure may not be available in a baseline survey. Therefore, this method allows us to apply a CiC strategy to nonrepeatable binary outcome measures, such as survival events and degree attainment. As the DiD strategy is not available without identically comparable repeated outcomes, this feature has added value to the CiC strategies themselves. Third, I proposed estimators that have several desirable properties, including flexible covariates, outcome specifications, and clustered sampling robustness.

Under the exclusion restriction requiring that the assignment's direct effect be 0, the ITT conditional on the treated, which I termed the ITTTA, must equal the ATT. In the microcredit application, I found that the conventional IV estimate for the ITTTA/ATT is 2.2 times larger than that of my preferred estimate, while the estimates are not statistically different. As the estimates are not distinguishable due to large standard errors, my empirical result does not offer a conclusive statement on the validity of the conventional estimate. In addition, my comparison may not be fair because the focus of my procedure was limited to a subsample of business owners. Nevertheless, this study highlights the importance of directly identifying treatment effect heterogeneity. Overall, I provide an option for applied researchers to further understand treatment effects and accordingly offer better policy guidelines.

By allowing nonrepeated measures within CiC strategies, additional research questions arise for the identification of treatment effects. For example, the need for a continuous proxy is a major limitation of the proposed procedure, emerging as the limited interpretability of the results in this study. Overall, continuous baseline variables are not always available. Thus, identification without a continuous proxy is an important future research topic. The estimation procedure followed a specific semiparametric distribution regression to tackle the possible high dimensionality. However, the proposed procedure may have resulted in a small sample problem when many discrete covariates were present. Testing rank similarity restrictions is also an important future research direction. In Appendix D.4, I implemented a version of \citeauthor{DongShen18}' (\citeyear{DongShen18}) means test for rank similarity through baseline covariates. Although the diagnostic test is successful, more sophisticated rank similarity tests that require additional data and restrictions are desirable in this context. In particular, for one-sided noncompliance experiments, the endogenous taking up of a treatment will not be observed in a control group and this produces tougher challenges for testing rank similarity. Overall, undertaking these identification and estimation issues remains an important topic for future research.

\newpage

\appendix

\begin{centering}
 \Large
 \textbf{Supplementary Appendix}
\normalsize
\end{centering}

\tableofcontents

\section{Asymptotics for the proposed estimators}
\subsection{For random sample of individual observations.}

I first assume the data generating process satisfies the following restrictions:
\begin{assumption}[DGP] \label{ass:condDGP}
 The sample $\{Y_i,Y_{b,i},D_i,W_i,T_i\}_{i=1}^{n}$ is an iid draw from the
 probability law $P$ over the support $\{\mathcal{Y} \times \mathcal{Y}_b \times \{0,1\}
 \times \mathcal{W} \times \{0,1\}\}$. Let $\mathcal{Y}^{0}$ be a support of $Y$ conditional on $T = 0$, and
 let $\mathcal{Y}^{1,d}_b$ and $\mathcal{W}_{d}$ be supports of $Y_b$ and
 $W$ conditional on $T = 1$ and $D = d$ for each $d \in
 \{0,1\}$. 

 Assume the following
 \begin{enumerate}
  \item $\mathcal{Y}^0 \times
 \mathcal{W}$ and $\mathcal{Y}^{1,d}_b \times \mathcal{W}_d$ are compact
 subsets of $\mathbb{R}^{1+d_w}$ for each $d \in
 \{0,1\}$.
  \item If $Y(0)$ is absolutely continuous with respect to the
 Lebesgue measure, then suppose the conditional density
 $f_{Y|W,T}(y_0|w,0)$ is uniformly bounded and
 uniformly continuous in $(y_0,w) \in \mathcal{Y}^0 \times
 \mathcal{W}$.
  \item $f_{Y_b|W,T}(y_b|w,1)$ and
 $f_{Y_b|W,T,D}(y_b|w,1,d)$ are uniformly bounded, and uniformly continuous
 in  and $(y_b,w) \in \mathcal{Y}_b \times \mathcal{W}$ for each $d \in
 \{0,1\}$.
  \item $\frac{n_{1,d}}{n}
 \equiv \frac{1}{n}\sum_{i} I\{T_i = 1,D_i=d\} \rightarrow^{p}
 \alpha_{1,d} \equiv Pr(T_i = 1,D_i = d) >
 0$ for each $d \in \{0,1\}$, and $\frac{n_0}{n} =
 \frac{1}{n}\sum_i I\{T_i = 0\} \rightarrow^p \alpha_0 \equiv Pr(T_i = 0) > 0$.
 \end{enumerate}
\end{assumption}

I also assume that the conditional distribution functions have the
following semiparametric forms
\begin{assumption}[Distribution Regression] \label{ass:QRFR}
 Suppose we have
 \[
  F_{Y|W,T}(y|w,0) = \Lambda(w'\beta^0(y)),
 \]
 for some link function $\Lambda(\cdot)$ for all $y, w$. For this
 specification, assume that the minimal eigenvalue of 
 \[
  J_{0}(y) \equiv E\left[\frac{\lambda(W'\beta^{0}(y))^2}{\Lambda(W'\beta^0(y)) [1 -
 \Lambda(W'\beta^0(y))]}W W' \right]
 \]
 is bounded away from zero uniformly over $y$, where $\lambda$ is the
 derivative of $\Lambda$. Assume also that the analogous restriction holds for
 \[
  F_{Y_b|W,T}(y|w,1) = \Lambda(w'\beta^{1}(y))
 \]
 and 
 \[
  F_{Y_b|W,T,D}(y|w,1,d) = \Lambda(w'\beta^{1,d}(y))
 \]
 for each $d \in \{0,1\}$.

 Assume further that $E\|W\|^{2} < \infty$.
\end{assumption}
This is a standard regularity condition for distribution regression
models (\citealp{ChernozhukovFernandezValMelly13}).


Under these assumptions, these conditional distribution functions weakly converge jointly. Let
 \begin{align*}
  \hat{G}^{1,d}(y_{b}^{1,d},w) =&
  \sqrt{n}\left(\hat{F}_{Y_{b}|W,T,D}(y_{b}^{1,d}|w,1,d) -
  F_{Y_b|W,T,D}(y_{b}^{1,d}|w,1,d)\right), \forall y_{b}^{1,d} \in \mathcal{Y}^{w,\{1,d\}}_{b}\\
  \hat{G}^1(y_b^1,w) =& \sqrt{n}\left(\hat{F}_{Y_b|W,T}(y_b^1|w,1) -
  F_{Y_b|W,T}(y_b^1|w,1)\right), \forall y_b^1 \in \mathcal{Y}^{w,1}_b \\
  \hat{G}^{0}(y_0,w) =& \sqrt{n}\left(\hat{F}_{Y(0)|W,T}(y_0|w,0) -
  F_{Y(0)|W,T}(y_0|w,0)\right), \forall y_0 \in \mathcal{Y}^{w,0},
 \end{align*}
 for every $w,d \in \mathcal{W}
 \times \{0,1\}$.
\begin{lemma} \label{lmm:wkConv} 
 Under assumptions for Theorem 3.2, and, Assumptions \ref{ass:condDGP} and \ref{ass:QRFR}, 
 \[
  \left(\hat{G}^{1}(y_{b}^1,w),\hat{G}^{1,d}(y_{b}^{1,d},w),\hat{G}^{0}(y_0,w)\right)
 \rightsquigarrow \left(\mathbb{G}^{1}(y_{b}^1,w),\mathbb{G}^{1,d}(y_{b}^{1,d},w),\mathbb{G}^0(y_0,w) \right)
 \]
 in $l^{\infty}(\mathcal{Y}_b\times\mathcal{W} \times
 \mathcal{Y}_b^{1,d}\times\mathcal{W}_d \times
 \mathcal{Y}^0\times\mathcal{W})$, where $\mathbb{G}^k(y,w)$ for every
 $k \in \{1,\{1,d\}_{d \in \{0,1\}},0\}$ are tight zero-mean Gaussian processes with
 each covariance function of the form
 \begin{align*}
  V_{k,k}(y,w,\tilde{y},\tilde{w})&
  = \alpha_{k}^{-1} w'J_{k}^{-1}(y)\lambda_k(w'\beta^k(y)) \Sigma_k(y,\tilde{y})
  \lambda_{k}(\tilde{w}'\beta^{k}(\tilde{y}))
  J_{k}^{-1}(\tilde{y}) \tilde{w}\\
  V_{1,\{1,d\}}(y,w,\tilde{y},\tilde{w})&
  = \alpha_{1,d}^{-1} w'J_{1}^{-1}(y)\lambda_1(w'\beta^1(y)) \Sigma_{1,\{1,d\}}(y,\tilde{y})
  \lambda_{1,d}(\tilde{w}'\beta^{1,d}(\tilde{y}))
  J_{1,d}^{-1}(\tilde{y}) \tilde{w}
 \end{align*}
 where
 \begin{align*}
    \Sigma_k(y,\tilde{y}) =& E[I\{K = k\}WH(W'\beta^k(y))\\
  \times &\{\min\{\Lambda(W'\beta^k(y)),\Lambda(W'\beta^k(\tilde{y}))\} -
  \Lambda(W'\beta^k(y))\Lambda(W'\beta^k(\tilde{y}))\}\\
  & \times H(W'\beta^k(\tilde{y}))]W',\\
  \Sigma_{1,\{1,d\}}(y,\tilde{y}) =& E[I\{T = 1,D = d\}WH(W'\beta^1(y))\\
  \times &\{\min\{\Lambda(W'\beta^1(y)),\Lambda(W'\beta^{1,d}(\tilde{y}))\} -
  \Lambda(W'\beta^{1}(y))\Lambda(W'\beta^{1,d}(\tilde{y}))\}\\
  & \times H(W'\beta^{1,d}(\tilde{y}))W']
 \end{align*}
 for each $k \in \{1,\{1,d\}_{d \in \{0,1\}},0\}$
 and $V_{1,0}(y,w,\tilde{y},\tilde{w}) = V_{\{1,d\},0}(y,w,\tilde{y},\tilde{w}) = 0$.

\end{lemma}

\begin{proof}[Proof of Lemma \ref{lmm:wkConv}]
 {\it Step 1: Weak convergences of Z-functions}

For the proof of Lemma \ref{lmm:wkConv}, we
would like to introduce approximate Z-map notations.

 For every $y \in \mathcal{Y}$, let $\Psi(y,\beta)$ be $dW$-vector of
 population moment equations such
 that the true parameter $\beta^* \in \mathbb{R}^{dW}$ solves the moment condition
 $\Psi(y,\beta) = 0$. Let $\hat{\Psi}(y,\beta)$ be an empirical
 analogue. Let an estimator $\hat{\beta}(y)$ satisfies
 \[
  \|\hat{\Psi}(y,\hat{\beta}(y))\|^2 \leq \inf_{\beta \in
 \mathbb{R}^{dW}} \|\hat{\Psi}(y,\beta)\|^2 + \hat{r}(y)^2
 \]
where $\hat{r}(y)$ is a numerical tolerance parameter with
 $\|\hat{r}\|_{\mathcal{Y}} = o_p(n^{-1/2})$. 

 Let $\phi(\Psi(y,\cdot),r(y)):
 l^{\infty}(\mathbb{R})^{dW} \times \mathbb{R} \mapsto \mathbb{R}^{dW}$
 be an approximate Z-map which assigns one of its $r(y)$-approximate zeros to each element
 $\Psi(y,\cdot)$ so that
 \[
  \hat{\beta}(\cdot) = \phi(\hat{\Psi}(\cdot,\cdot),\hat{r}(\cdot)),
 \beta^*(\cdot) = \phi(\Psi(\cdot,\cdot),0).
 \]

 For each $q = 1,\ldots, dW$, let 
 \[
  \psi^{0,q}_{y,\beta^0}(Y,W,T) = I\{T = 0\}\left[\Lambda(W'\beta^0) - I\{Y
 \leq y\} \right]H(W'\beta^0)W_q,
 \]
 \[
  \psi^{1,q}_{y,\beta^1}(Y,W,T) = I\{T = 1\}\left[\Lambda(W'\beta^1) - I\{Y_b
 \leq y\} \right]H(W'\beta^1)W_q,
 \]
 \[
  \psi^{\{1,d\},q}_{y,\beta^{1,d}}(Y,W,T,D) = I\{T = 1,D = d\}
 \left[\Lambda(W'\beta^{1,d}) - I\{Y_b \leq y\}
 \right]H(W'\beta^{1,d})W_q,
 \]
 and 
 for each $d \in \mathcal{D}_1$. Also let $\psi^{0}_{y,\beta^0}$, $\psi^{1}_{y,\beta^1}$ and
 $\psi^{1,d}_{y,\beta^{1,d}}$ be $dW$-vector valued functions with each
 q-th coordinate being $\psi^{0,q}_{y,\beta^0}$, $\psi^{1,q}_{y,\beta^1}$ and $\psi^{\{1,d\},q}_{y,\beta^{1,d}}$.

 In Lemma \ref{lmm:Donsker} below, it is shown
 that the union of classes of functions
 \begin{align*}
  \cup_{d \in \mathcal{D}_{1}, q \in \{1,\ldots, dW\}} &\{\psi^{0,q}_{y,\beta^0}(Y,W,T): (y,\beta^0) \in \mathcal{Y} \times
  \mathbb{R}^{dW})\}\\
  &\cup \{\psi^{1,q}_{y,\beta^1}(Y,W,T): (y,\beta^1) \in \mathcal{Y} \times
  \mathbb{R}^{dW})\} \\
  &\cup \{\psi^{\{1,d\},q}_{y,\beta^{1,d}}(Y,W,T): (y,\beta^{1,d}) \in \mathcal{Y} \times \mathbb{R}^{dW})\},
 \end{align*}
 is $P$-Donsker with a square-integrable envelope function. Let
 $\hat{\Psi}^t = P_n \psi^t_{y,\beta^t}, \hat{\Psi}^{1,d} = P_n \psi^{1,d}_{y,\beta^{1,d}}$ and
 $\Psi^t = P \psi^t_{y,\beta^t}, \Psi^{1,d} = P \psi^{1,d}_{y,\beta^{1,d}}$ for each $t \in
 \{0,1\}$, then the Donskerness implies
 \[
  \left(\sqrt{n}(\hat{\Psi}^{1} - \Psi^1),\sqrt{n}(\hat{\Psi}^{0} -
 \Psi^0),\sqrt{n}(\hat{\Psi}^{1,d} - \Psi^{1,d})\right) \rightsquigarrow (\mathbb{G}(\psi_{y^1,\beta^{1}}^{1}),\mathbb{G}(\psi_{y^0,\beta^{0}}^{0}),\mathbb{G}(\psi_{y^{1,d},\beta^{1,d}}^{1,d}))
 \]
 in $l^{\infty}(\mathcal{Y}^1_b \times \mathbb{R}^{dW})^{dW} \times
 l^{\infty}(\mathcal{Y}^0 \times \mathbb{R}^{dW})^{dW} \times
 l^{\infty}(\mathcal{Y}^{1,d}_b \times \mathbb{R}^{dW})^{dW}$
 where $\mathbb{G}(\psi_{y,\beta^{k}}^{k})$ for each $k \in
 \{1,0,\{1,d\}\}$ are $P$-Brownian bridges.

 {\it Step 2: Applying Functional Delta method through the stacking rule}

 From the first order conditions, $\hat{\beta}^k(y) =
 \phi^k(\hat{\Psi}^k(y,\cdot),\hat{r}^k(y)), \hat{r}^k(y) = \max_{1 \leq
 i \leq n}\|W_i\|dW/n$ for each $y \in \mathcal{Y}^k$ and
 $n^{1/2}\|\hat{r}\|_{\mathcal{Y}^k} \rightarrow^{\mathbb{P}} 0$, and
 $\beta_k(y) = \phi^k(\Psi^k(y,\cdot),0)$ for each $y \in
 \mathcal{Y}^k$ for every $k \in \{1,0,\{1,d\}\}$.

 Following the argument of \cite{ChernozhukovFernandezValMelly13},
 the three kinds of approximate Z-maps $\phi^{1},\phi^{0},\phi^{1,d}$ are
 Hadamard differentiable for each case, and from the stacking rule as in
 Lemma B.2 of \cite{ChernozhukovFernandezValMelly13}, we have 
 \begin{align*}
  &\left(\sqrt{n}(\hat{\beta}^{1}(y^1) -
 \beta^1(y^1)),\sqrt{n}(\hat{\beta}^{0}(y^0) -
 \beta^0(y^0)),\sqrt{n}(\hat{\beta}^{1,d}(y^{1,d}) - \beta^{1,d}(y^{1,d})) \right)\\
 &\rightsquigarrow \left(-J_1^{-1}\mathbb{G}(\psi_{y^1,\beta^1(y^1)}^{1}),-J_0^{-1}\mathbb{G}(\psi_{y^0,\beta^0(y^0)}^{0}),-J_{1,d}^{-1}\mathbb{G}(\psi_{y^{1,d},\beta^{1,d}(y^{1,d})}^{1,d}) \right)
 \end{align*}
 in $l^{\infty}(\mathcal{Y}^1_b)^{dW} \times l^{\infty}(\mathcal{Y}^0)^{dW} \times
 l^{\infty}(\mathcal{Y}^{1,d}_b)^{dW}$ by the functional delta method.

 {\it Step 3: Applying another Hadamard differentiable map to conclude
 the statement}

 Finally, consider the mapping $\nu^{k}: \mathbb{D}_{\nu^{k}} \subset
 l^{\infty}(\mathcal{Y}^k)^{dW} \mapsto l^{\infty}(\mathcal{Y}^k \times
 \mathcal{W})$ such that
 \[
  b \mapsto \nu^{k}(b), \nu^k(b)(w,y) = \Lambda(w'b(y))
 \]
 for every $k \in \{1,0,\{1,d\}\}$. From the Hadamard differentiability
 of $\nu^k$
 at $b^k(\cdot) = \beta^k(y)$ tangentially to $C(\mathcal{Y}^k)^{dW}$
 with the derivative map $\alpha \mapsto \nu_{\beta^k(\cdot)}'(\alpha)(w,y) =
 \lambda(w'\beta^k(y))w'\alpha(y)$. From the stacking rule, applying the
 mapping for each process, the statement of the lemma holds.
\end{proof}

\begin{lemma} \label{lmm:Donsker}
 Under the assumptions of Lemma \ref{lmm:wkConv}, the class of
 functions 
 \begin{align*}
  \cup_{d \in \mathcal{D}_{1}, q \in \{1,\ldots, dW\}} &\{\psi^{0,q}_{y,\beta^0}(Y,W,T): (y,\beta^0) \in \mathcal{Y}^0 \times
  \mathbb{R}^{dW})\}\\
  &\cup \{\psi^{1,q}_{y,\beta^1}(Y,W,T): (y,\beta^1) \in \mathcal{Y}_b^1 \times
  \mathbb{R}^{dW})\} \\
  &\cup \{\psi^{\{1,d\},q}_{y,\beta^{1,d}}(Y,W,T): (y,\beta^{1,d}) \in \mathcal{Y}_{b}^{1,d} \times \mathbb{R}^{dW})\},
 \end{align*}
 is $P$-Donsker with a square-integrable envelope.
\end{lemma}
\begin{proof}
 From Theorem 19.14 in \cite{vanderVaart98}, a suitable measurable class
 of measurable functions $\mathcal{G}$ is $P$-Donsker if the uniform
 entropy integral with respect to an envelope function $G$
 \[
  J(1,\mathcal{G},L_2) = \int_0^1 \sqrt{\log \sup_Q N_{[]}(\epsilon
 \|G\|_{Q,2}, \mathcal{G},L_2(G))} d\epsilon
 \]
 is finite and the envelope function $G$ satisfies $\mathbb{P}_1 G^2 < \infty$.

 Consider classes of functions
  \[
  \mathcal{F}_1 = \{W'\beta: \beta \in \mathbb{R}^{dW}\}, \mathcal{F}_{2,k}
 = \{ I\{Y \leq y\}, y \in \mathcal{Y}^k\}, \{W_q:q = 1,\ldots , d_w\}
 \]
 which are VC classes of functions. Note that the target class of
 functions is the union of 
 \[
  \{I\{T = t\}(\Lambda(\mathcal{F}_1) - \mathcal{F}_{2,t})H(\mathcal{F}_1)W_q : q = 1, \ldots, d_w\}
 \]
 for each $t \in \{0,1\}$ and 
 \[
  \{I\{T = 1,D = d\}(\Lambda(\mathcal{F}_1) -
 \mathcal{F}_{2,\{1,d\}})H(\mathcal{F}_1)W_q : q = 1, \ldots, d_w\}
 \]
 for each $d \in \mathcal{D}_1$. These are Lipschitz transformation of
 VC-class of functions and finite set of functions $I\{T = t\}$ and
 $I\{T = 1,D = d\}$ where the Lipschitz coefficients bounded by $const \cdot
 \|W\|$. Therefore, from Example 19.19 of \cite{vanderVaart98}, the
 constructed class of functions has the finite uniform entropy integral relative to the envelope
 function $const \cdot \|W\|$, which is square-integrable from the
 assumption. Suitable measurability is granted as it is a pointwise measurable class of functions. Thus, the class of functions is Donsker.
\end{proof}

\begin{lemma}\label{lmm:DKP}
 Under assumptions of Theorem 3.2, and Assumption
 \ref{ass:condDGP}, it is possible to construct a class of measurable
 functions $\mathcal{F}$ including
 \[
  \{F_{Y|W,K}(y|\cdot,k), y \in \mathcal{Y}^k, k \in K \equiv \{\{T = 1\},\{T = 0\},\{T = 1,D = d\}_{d \in \mathcal{D}_1}\}
 \]
 and all the indicators of the rectangles in $\bar{\mathbb{R}}^{dW}$
 such that $\mathcal{F}$ is DKP
 class (\citealp{ChernozhukovFernandezValMelly13}, Appendix A.).

\end{lemma} 
 
\begin{proof}
 As in Step 2 in the proofs of Theorem 5.1 and 5.2 in
 \cite{ChernozhukovFernandezValMelly13}, $\mathcal{F}_1 =
 \{F_{Y|W,T}(y|\cdot,1), y \in \mathcal{Y}^1_b\},\mathcal{F}_0 =
 \{F_{Y|W,T}(y|\cdot,0), y \in \mathcal{Y}^0\}$ and $\mathcal{F}_{1,d} =
 \{F_{Y|W,T,D}(y|\cdot,1,d), y \in \mathcal{Y}_b^{1,d}\}$ are uniformly
 bounded ``parametric'' family (Example 19.7 in \cite{vanderVaart98})
 indexed by $y \in \mathcal{Y}^{k}$ for
 each $k \in \{1,0,\{1,d\}\}$ respectively. From the assumption that the
 density function $f_{Y|W,K}(y|\cdot,k)$ is uniformly bounded,
 \[
  |F_{Y|W,K}(y|\cdot,k)I\{K = k\} - F_{Y|W,K}(y'|\cdot,k)I\{K = k\}| \leq L |y - y'|
 \]
 for some constant $L$ for every $k \in \{1,0,\{1,d\}\}$. The
 compactness of $\mathcal{Y}^k$ implies the uniform $\epsilon$-covering
 numbers to be bounded by $const/\epsilon$ independent of
 $F_{W}$ so that the Pollard's entropy condition is met. Therefore, the class of $\mathcal{F}$ is suitably
 measurable as well. As indicator
 functions of all rectangles in $\bar{R}^{dW}$ form a VC class, we can
 construct $\mathcal{F}$ that contains union of all the families
 $\mathcal{F}_0, \mathcal{F}_1$ and $\mathcal{F}_{1,d}$ and the
 indicators of all the rectangles in $\bar{\mathbb{R}}^{dW}$ that
 satisfies DKP condition.


\end{proof}

Given the weak convergence of the distribution regressions, the conditional estimator
\[
 \hat{F}_{Y_b|W,T,D}(\hat{Q}_{Y_b|W,T}(\hat{F}_{Y(0)|W,T}(y|w,0)|w,1)|w,1,d)
\]
is in the form of following map: for distribution functions
$F^b_{1,d},F^b_{1},F^{0}$,
\[
 m(F^b_{1,d},F^b_1,F^{0}) = F^b_{1,d} \circ Q^b_{1} \circ F^{0}.
\]
where $Q^b_{1}$ is the quantile function from $F^b_1$.
In a parallel argument to \cite{MellySantanglo15} based on quantile
regressions, the above map is Hadamard
differentiable from the Hadamard-differentiability of the quantile
function (Lemma 21.4 (ii), \citealp{vanderVaart98}) and the chain
rule of the Hadamard-differentiable maps (Lemma 20.9,
\citealp{vanderVaart98}).
\begin{lemma}\label{lmm:HadamardMap}
 Let $F^b_{1}$ and $F^b_{1,d}$ be uniformly continuous and differentiable
 distribution functions with uniformly bounded densities $f^b_1$ and
 $f^b_{1,d}$. Let $F^{0}$ be also a distribution function. Suppose $F^b_{1}$
 has a support $[a,b]$ as a bounded subset of real line, and $F^b_1 \circ
 Q^b_1(p) = p$ for every $p \in [0,1]$.

 Then the map $m(F^b_{1,d},F^b_1,F^{0})$ is Hadamard differentiable at
 $(F^b_{1,d},F^b_{1},F^{0})$ tangentially to a set of functions
 $h^b_{1,d},h^b_{1},h^0$ with the derivative map
 \[
  h^b_{1,d} \circ Q^b_1 \circ F^0 - f^b_{1,d}(Q^b_1 \circ F^0) \frac{h^b_{1}
 \circ Q^b_1 \circ F^0}{f^b_{1} (Q^b_1 \circ F^0)} +
 f^b_{1,d}(Q^b_1 \circ F^0) \frac{h^0}{f^b_{1} (Q^b_1 \circ F^0)}.
 \]
\end{lemma}
For the proof of the Hadamard derivative expression, I show the
following lemma.
\begin{lemma} \label{lmm:inverseDeriv}
 Suppose $F_1$ is a uniformly continuous and differentiable distribution
 function with uniformly bounded density function $f_1$ with its support
 $[a,b]$ where $0 < a < b < 1$. Let $\psi_{F_1}(p) = Q_1(p)$ for all $p \in
 [0,1]$.

 Let $F_0$ be a distribution function over a support $\mathcal{Y}^0$,
 then $\psi_{F_1}$ is Hadamard differentiable at $F_0$
 tangentially to a set of function $h_0$ with the derivative map
 \[
  \frac{h_0}{f_1 \circ Q_1 \circ F_0}.
 \] 
\end{lemma}
\begin{proof}
 From the assumption, we have
 \[
  F_1 \circ Q_1 \circ (F_0 + t h_{0t}) - F_1 \circ Q_1 \circ F_0 =
 (F_0 + t h_{0t}) - F_0 = t h_{0t}.
 \]
 Let $h_{0t} \rightarrow h_0$ uniformly in $\mathcal{Y}_0$ as $t
 \rightarrow 0$, and let $g_t \equiv \psi_{F_1} \circ (F_0 + h_{0t})$
 and $g \equiv \psi_{F_1} \circ F_0$. Then $g_t \rightarrow g$ uniformly
 in $\mathcal{Y}_0$ as $t \rightarrow 0$ by the assumption.

 Thus, we have
 \[
  \frac{F_1(g_t) - F_1(g)}{g_t - g} \frac{g_t - g}{t} = h_{0t}
 \]
 so that we have
 \[
  \frac{g_t - g}{t} = \frac{h_{0t}}{\frac{F_1(g_t) - F_1(g)}{g_t - g}}
 \]
 whereas the RHS has a limit
 \[
  \frac{h_{0t}}{\frac{F_1(g_t) - F_1(g)}{g_t - g}} \rightarrow^{t
 \rightarrow 0} \frac{h_0}{f_1 \circ Q_1 \circ F_0}
 \]
 by the uniform differentiability of $F_1$.

\end{proof}

\begin{proof}[Proof of Lemma \ref{lmm:HadamardMap}]
 First, let $\phi(F_1,F_0) \equiv Q_1 \circ F_0$ so that
 $m(F_{1,d},F_1,F_0) = F_{1,d} \circ \phi(F_1,F_0)$. From the lemma
 \ref{lmm:inverseDeriv} and the lemma 21.4 (ii) from
 \cite{vanderVaart98}, $\phi$ is Hadamard differentiable at $(F_1,F_0)$
 tangentially to the set of functions $(h_1,h_0)$ with the derivative
 map
 \[
  - \frac{h_1 \circ Q_1 \circ F_0}{f_1 \circ Q_1 \circ F_0} +  \frac{h_0}{f_1 \circ Q_1 \circ F_0}.
 \]
 Therefore, from the Chain rule of the Hadamard differentiability (Lemma
 20.9, \citealp{vanderVaart98}), the map $m$ is Hadamard differentiable
 with the derivative map shown in the lemma.
\end{proof}

 Next, let
 \[
  \hat{G}(f) = \sqrt{n}\left(\int f d\hat{F}_{W,T,D} - \int f dF_{W,T,D} \right)
 \]
 where $\hat{F}_{W,T,D}(w,t,d) = n^{-1}\sum_{i=1}^n I\{W_i \leq w, T =
 t, D = d\}$ for $f \in \mathcal{F}$ where $\mathcal{F}$ is a class of
 suitably measurable functions \footnote{Suitably measurability or
 $P$-measurability can be verified by showing the class is pointwise
 measurable.
 
 A class $\mathcal{F}$ of measurable functions is pointwise measurable if
there is a countable subset $\mathcal{G} \subset \mathcal{F}$ such that
for every $f \in \mathcal{F}$ there is a sequence $\{g_m\} \in
\mathcal{G}$ such that
\[
 g_m(x) \rightarrow f(x)
\]
for every $x \in \mathcal{X}$.

For example, the class $\mathcal{F} \equiv \{1 \{ x \leq t\}, t \in
\mathbb{R}\}$ is pointwise measurable. This is because we can take 
\[
 \mathcal{G} = \{I\{x \leq t\}: t \in \mathbb{Q}\}
\]
which is countable, and for arbitrary $t_0 \in \mathbb{R}$ which
characterise the arbitrary function $f(x) = \{I\{x \leq t_0\}\}$, and
the sequence of functions
\[
 g_m(x) = I\{x \leq t_m\}
\]
such that $t_m \geq t_0, t_m \rightarrow t_0$ converges to
$f(x)$. Therefore, the relevant class used here is shown to be pointwise
measurable.} including 
 \[
  \{F_{Y|W,K}(y|\cdot,k), y \in \mathcal{Y}^k, k \in K \equiv \{\{T = 1\},\{T = 0\},\{T = 1,D = d\}_{d \in \{0,1\}}\}
 \]
 and all the indicators of the rectangles in $\bar{\mathbb{R}}^{dW}$.  

From the derivative expression in the lemma \ref{lmm:HadamardMap} and the joint weak
 convergence of the empirical processes, the counterfactual conditional
 distribution weakly converges.
\begin{theorem}
 Under assumptions for Theorem 3.2, and, Assumptions \ref{ass:condDGP} and \ref{ass:QRFR}, 
 \[
  \hat{G}(f) \rightsquigarrow \mathbb{G}(f)
 \]
 in $l^{\infty}(\mathcal{F})$ for $\mathcal{F}$ specified earlier where
 $\mathbb{G}(f)$ is a Brownian bridge, and
 \[
   \sqrt{n}\left(\hat{F}_{Y(0)|W,T,D}(y|w,1,d) - F_{Y(0)|W,T,D}(y|w,1,d) \right)
 \rightsquigarrow \mathbb{G}^{WCF}_{1,d}(y,w), \mbox{ in }
 l^{\infty}(\mathcal{Y}^0 \times \mathcal{W})
 \]
 where $\mathbb{G}^{WCF}_{b,d}(y,w)$ is a tight zero mean Gaussian
 process indexed by $(y,w)$ such that
 \begin{align*}
  \mathbb{G}^{WCF}_{1,d}&(y,w)\\
  =& \mathbb{G}^{1,d}(Q_{Y_b|W,T}(F_{Y(0)|W,T}(y|w,0)|w,1),w)\\
 -& \frac{f^b_{1,d}(y,w)}{f^b_1(y,w)}
  \left(\mathbb{G}^{1}(Q_{Y_b|W,T}(F_{Y(0)|W,T}(y|w,0)|w,1),w) +
  \mathbb{G}^{0}(y,w)\right).
 \end{align*}
 where 
 \[
  \frac{f^b_{1,d}(y,w)}{f^b_1(y,w)} \equiv
 \frac{f_{Y_b|W,T,D}(Q_{Y_b|W,T}(F_{Y(0)|W,T}(y|w,0)|w,1)|w,1,d)}{f_{Y_b|W,T}(Q_{Y_b|W,T}(F_{Y(0)|W,T}(y|w,0)|w,1)|w,1)}
 \]
\end{theorem}
\begin{proof}

 From Lemma \ref{lmm:DKP}, we can choose $\mathcal{F}$
 satisfying the requirement defined earlier so that $\mathcal{F}$ satisfies the
 DKP condition (\citealp{ChernozhukovFernandezValMelly13}, Appendix
 A.). Then assumptions of Lemma E.4 in
 \cite{ChernozhukovFernandezValMelly13} is satisfied to conclude the
 first statement.
 
 For the second statement, note that
 \begin{align*}
  \sqrt{n}&\left(\hat{F}_{Y(0)|T,W,D}(y|1,w,d) - F_{Y(0)|T,W,D}(y|1,w,d) \right)\\
  &= \sqrt{n}\left(m(\hat{F}_{Y_b|W=w,D=d},\hat{F}_{Y_b|W=w},\hat{F}_{Y(0)|W=w})
  - m(F_{Y_b|W=w,D=d},F_{Y_b|W=w},F_{Y(0)|W=w})\right).
 \end{align*}
 
 Therefore, the functional delta method and the Hadamard
 differentiability of the transformation $m(\cdot,\cdot,\cdot)$
 implies the above process weakly converges to the process shown in the
 statement.
\end{proof}

The unconditional counterfactual distribution is attained by applying
Lemma D.1 from \cite{ChernozhukovFernandezValMelly13} showing the
Hadamard differentiability of the counterfactual operator
\[
 \phi^{C}(F,G) = \int F(y|w) dG(w)
\]
\begin{theorem}
 Under assumptions for Theorem 3.2, and, Assumptions \ref{ass:condDGP} and \ref{ass:QRFR}, 
 \[
   \sqrt{n}\left(\hat{F}_{Y(0)|T,D}(y|1,d) - F_{Y(0)|T,D}(y|1,d) \right)
 \rightsquigarrow \mathbb{G}^{CF}_{1,d}(y), \mbox{ in } l^{\infty}(\mathcal{Y}_0).
 \]
 where $\mathbb{G}^{CF}_{1,d}(y)$ is a tight mean zero Gaussian
 process such that
 \begin{align*}
  \mathbb{G}^{CF}_{1,d}(y) \equiv &\alpha_{1,d}^{-1}\int \mathbb{G}_{1,d}^{WCF}(y,w)I\{T =
 1,D = d\} dF_{W,T,D}(w,1,d)\\
 &+ \alpha_{1,d}^{-1}\mathbb{G}(F_{Y(0)|W,T,D}(y|w,1,d)I\{T = 1,D = d\}).  
 \end{align*}
\end{theorem}
\begin{proof}
 Note that
 \[
  \sqrt{n}\hat{F}_{Y(0)|T,D}(y|1,d) = \sqrt{n}\frac{\hat{F}_{Y(0),T,D}(y,1,d)}{n_{1,d}/n}
 \]
 and
 \[
  \sqrt{n}F_{Y(0)|T,D}(y|1,d) = \sqrt{n}\frac{F_{Y(0),T,D}(y,1,d)}{\alpha_{1,d}}.
 \]
 Also we have,
 \begin{align*}
  \sqrt{n}&(\hat{F}_{Y(0),T,D}(y,1,d) - F_{Y(0),T,D}(y,1,d))\\
  =& \sqrt{n}\left(\frac{1}{n}\sum_{i=1}^n
  \hat{F}_{Y(0)|W,T,D}(y|W_i,1,d)I\{T_i = 1, D_i = d\} - \int
  F_{Y(0)|W,T,D}(y|W,1,d)I\{T = 1,D = d\} dP\right)\\
  =& \sqrt{n}\left(\int 
  \hat{F}_{Y(0)|W,T,D}(y|W,1,d)I\{T = 1, D = d\} dP_n  - \int
  F_{Y(0)|W,T,D}(y|W,1,d)I\{T = 1,D = d\} dP\right).
 \end{align*}

 From the Slutzky lemma (Theorem 18.10 (v) in \citealp{vanderVaart98}),
 Functional delta method, and the Hadamard differentiability of the operator $\phi^C(F,G)$
 (\citealp{ChernozhukovFernandezValMelly13}, Lemma D.1), we have
 \begin{align*}
  &\left(\sqrt{n}(\hat{F}_{Y(0),T,D}(y,1,d) -
  F_{Y(0),T,D}(y,1,d)),n_{1,d}/n\right) \rightsquigarrow  \left(\tilde{\mathbb{G}}(y,w,d),\alpha_{1,d}\right)
 \end{align*}
 where
 \[
  \tilde{\mathbb{G}}(y,w,d) \equiv \left(\int \mathbb{G}_{1,d}^{WCF}(y,w) dF_{W,T,D}(w,1,d) +
 \mathbb{G}(F_{Y(0)|W,T,D}(y|\cdot,1,d)\right) I\{T =
  1,D = d\}
 \]

 Therefore, continuous mapping theorem (Theorem 18.11 in
 \citealp{vanderVaart98}) implies the statement.
\end{proof}

\begin{corollary}
 Under assumptions for Theorem 3.2, and, Assumptions \ref{ass:condDGP} and \ref{ass:QRFR}, 
 \[
   \sqrt{n}\left(\hat{Q}_{Y(0)|T,D}(y|1,d) - Q_{Y(0)|T,D}(y|1,d) \right)
 \rightsquigarrow
 -\frac{\mathbb{G}^{CF}_{1,d}(Q_{Y(0)|T,D}(y|1,d))}{f_{Y(0)|T,D}(Q_{Y(0)|T,D}(y|1,d)|1,d)},
 \mbox{ in } l^{\infty}(\mathcal{Y}_0)
 \] 
 for every $d \in \{0,1\}$.
\end{corollary}
\begin{proof}
 Immediate from the Hadamard-differentiability of the quantile function (Lemma 21.4 (ii), \citealp{vanderVaart98}).
\end{proof}

\begin{corollary}
 Under assumptions for Theorem 3.2, and, Assumptions \ref{ass:condDGP} and \ref{ass:QRFR}, 
 \[
  \sqrt{n}\left(\int y d\hat{F}_{Y(0)|T,D}(y|1,d) - \int y
 F_{Y(0)|T,D}(y|1,d)\right)\rightsquigarrow \int y d\mathbb{G}_{1,d}^{CF}(y)
 \]
\end{corollary}
\begin{proof}
 Let $\mu(F) = \int y dF(y)$ be the mapping $\mu: F \rightarrow
 \mathbb{R}$. Let $h_t \rightarrow h$ as $t \rightarrow 0$ and let $F_t
 = F + t h_t$. Then it is Hadamard differentiable at $F$ tangentially to
 a set of functions $h$ such that
 \[
  \frac{1}{t}\left(\int y dF_t - \int y dF \right) = \int y dh_t
 \rightarrow_{t \rightarrow \infty} \int ydh.
 \]
 Since
 \[
  \sqrt{n}(\hat{F}_{Y(0)|T,D}(y|1,d) - F_{Y(0)|T,D}(y|1,d))
 \rightsquigarrow \mathbb{G}_{1,d}^{CF}(y),
 \]
 the statement holds.
\end{proof}

Given the asymptotic normality, I propose an inference based on a
bootstrap procedure. Suppose the bootstrap draws are exchangeable.
\begin{assumption}[Exchangeable bootstrap] \label{ass:EXBt}
 Let $(w_{1},\ldots,w_n)$ is an exchangeable, non-negative random
 vector independent of the data $\{Y_{i},Y_{b,i},D_i,W_i,T_i\}_{i=1}^n$ such
 that for some $\epsilon > 0$,
 \[
  E[w_{1}^{2 + \epsilon}] < \infty, n^{-1} \sum_{i=1}^n (w_i -
 \bar{w})^2 \rightarrow^{\mathbb{P}} 1, \bar{w} \rightarrow^{\mathbb{P}} 1
 \]
 where $\bar{w} = n^{-1} \sum_{i=1}^n w_i$, and $\mathbb{P}$ is an outer
 probability measure with respect to $P$.\footnote{For an arbitrary
 maps $D: \Omega \mapsto \mathbb{D}$ on a metric space $\mathbb{D}$ and
 a bounded function $f:\mathbb{D} \mapsto \mathbb{R}$, 
 $\mathbb{P}f(D) = \inf \{PU : U:
 \Omega \mapsto \mathbb{R} \mbox{ is measurable }, U \geq f(D), PU
 \mbox{ exists.}\}$}
\end{assumption}

Let
$\hat{F}^{*}_{Y_b|W,T,D}(y|w,1,d),\hat{F}^{*}_{Y_b|W,T}(y|w,1),\hat{F}^{*}_{Y(0)|W,T}(y|w,0)$
be the bootstrapped version of the estimators using
\[
 \hat{\beta}^{*,k}(y) = \arg\max_{b \in \mathbb{R}^{dW}} \sum_{i=1}^n
 w_i I\{K_i = k\}\left[I\{Y_i \leq y\}\log \Lambda(W_i'b)  + I\{Y_i >
 y\}\log(1 - \Lambda(W_i'b))\right]
\]
and let
\[
 \hat{F}_{W,T,D}^*(w,1,d) = (n^*)^{-1} \sum_{i=1}^n w_i I\{W_i \leq
 w,T_i = 1, D_i = d\}, w \in \mathcal{W}
\]
where $n^* = \sum_i^{n} w_i$.
\begin{corollary} 
 Let
 \[
  \hat{G}^*(f) = \frac{1}{\sqrt{n}}\sum_{i=1}^n (w_{i} - \bar{w}) f
 \]
 for $f \in \mathcal{F}$, and let
 \[
  \hat{G}^{*,k}(y,w) = \sqrt{n}\left(\hat{F}^*_{Y|W,K}(y|w,k) - \hat{F}_{Y|W,K}(y|w,k) \right)
 \]
 Then under assumptions for Theorem 3.2, and, Assumptions \ref{ass:condDGP}, \ref{ass:QRFR}, and
 \ref{ass:EXBt}
 \[
  \left(\hat{G}^{*,1,d}(y,w),\hat{G}^{*,1}(y,w),\hat{G}^{*,0}(y,w),\hat{G}^*(f)\right)
 \rightsquigarrow^{\mathbb{P}} \left(\mathbb{G}^{1,d}(y,w),\mathbb{G}^1(y,w),\mathbb{G}^0(y,w),\mathbb{G}(f)\right)
 \]
 in $l^{\infty}(\mathcal{Y}_b^{1,d} \times \mathcal{W}_d \times
 \mathcal{Y}_b^{1} \times \mathcal{W} \times
 \mathcal{Y}_{0} \times \mathcal{W}) \times
 l^{\infty}(\mathcal{F})$. Therefore, 
 \[
  \sqrt{n}\left(\hat{F}^{*}_{Y(0)|T,D}(y|1,d) - \hat{F}_{Y(0)|T,D}(y|1,d) \right)
 \rightsquigarrow^{\mathbb{P}} \mathbb{G}^{CF}_{1,d}(y).
 \]
\end{corollary}
\begin{proof}
 Argument follows from Theorem 3.6.13 of \cite{vanderVaartWellner96} as
 employed in \cite{ChernozhukovFernandezValMelly13}, Theorem 5.1 and 5.2.
\end{proof}

\subsection{For clustered samples}
Now consider the asymptotic property of the clustered version of
estimator. First, I assume the following data generating process:
\begin{assumption}[Cluster DGP] \label{ass:condDGPcluster}
 Let $c = 1,\ldots, \bar{C}$ denote clusters, and an index $(i,c)$
 denote individual $i$'s observation in a cluster $c$. 

 Suppose the following restrictions:
 \begin{enumerate}
  \item  \[
	 \{N_c, \{Y_{i;c},Y_{b,i;c},D_{i;c},W_{i;c},T_c\}_{i \geq 1}\}_{c \geq 1}
	 \]
	 is exchangeable, namely, for any permutation $\pi$ of $\{1,\ldots,\bar{C}\}$,
	 \begin{align*}
	  &\{N_c, \{Y_{i;c},Y_{b,i;c},D_{i;c},W_{i;c},T_c\}_{i \geq 1}\}_{c \geq 1}\\
	  &\sim \{N_{\pi(c)},
	  \{Y_{i;\pi(c)},Y_{b,i;\pi(c)},D_{i;\pi(c)},W_{i;\pi(c)},T_{\pi(c)}\}_{i \geq 1}\}_{c
	  \geq 1}.  
	 \end{align*}
	 and clusters $\{N_c, \{Y_{i;c},Y_{b,i;c},D_{i;c},W_{i;c},T_c\}$ are
	 independent across $c$. 
  \item  $E[N_1] > 0, E[N_1^2] < \infty$ and $\bar{C} \rightarrow \infty$.
  \item $T$ is supported for $\{0,1\}$ and $D$ has a finite support $\{0,1\}$.
  \item the same support and density conditions as in Assumption \ref{ass:condDGP}.
  \item \[
	 \hat{\alpha}_{1,d} \equiv \frac{1}{\bar{C}}\sum_{1 \leq c \leq \bar{C}}\frac{1}{N_c}\sum_{i=1}^{N_c} I\{T_c =
 1,D_{i;c}=d\} \rightarrow \alpha_{1,d}
	\]
	where
	\[
	 \alpha_{1,d} \equiv E\left[\frac{1}{N_1}\sum_{i=1}^{N_1}
	I\{T_c = 1,D_{i;c} = d\} \right] > 0
	\]
	for each $d \in \{0,1\}$ , and $\hat{\alpha}_0 = \frac{1}{\bar{C}}\sum_{1 \leq c \leq
 \bar{C}} I\{T_c = 0\} \rightarrow \alpha_0
 \equiv E[I\{T_1 = 0\}] > 0$.
 \end{enumerate}
\end{assumption}
Furthermore, I modify the moment condition as follows:
\begin{assumption}[Distribution Regression and Cluster Moment Condition] \label{ass:QRFRcluster}
 Suppose we have
 \[
  F_{Y|W,T}(y|w,0) = \Lambda(w'\beta^0(y)),
 \]
 for some link function $\Lambda(\cdot)$ for all $y, w$. For this
 specification, assume that the minimal eigenvalue of 
 \[
  J_{0}(y) \equiv E\left[\sum_{i=1}^{N(1)}\frac{\lambda(W_i'\beta^{0}(y))^2}{\Lambda(W_i'\beta^0(y)) [1 -
 \Lambda(W_i'\beta^0(y))]}W_i W_i' \right]
 \]
 is bounded away from zero uniformly over $y$, where $\lambda$ is the
 derivative of $\Lambda$. Assume also that the analogous restriction holds for
 \[
  F_{Y_b|W,T}(y|w,1) = \Lambda(w'\beta^{1}(y))
 \]
 and 
 \[
  F_{Y_b|W,T,D}(y|w,1,d) = \Lambda(w'\beta^{1,d}(y))
 \]
 for each $d \in \{0,1\}$. Assume further that $E\left[N_1 \sum_{i=1}^{N_1} \|W_i\|^{2}\right] < \infty$.
\end{assumption}
\begin{remark}
 The additional assumption of $E[N_1 \sum_{i=1}^{N_1} \|W_i\|^2] <
 \infty$ is replacing the square integrability of the envelope function
 $const \cdot \|W_i\|$ for a class of Z-maps as the first order
 conditions of the semiparameteric distribution regressions. 

 From the linearity of the expectation operator, it is sufficient to
 have the cluster size finite $N_1 < \infty$ as well as the previous
 moment condition $E[\|W_i\|^2] < \infty$ for every $i \in N_1$.
\end{remark}

Let
\[
 \hat{G}^C(f) = \sqrt{\bar{C}}\left(\frac{1}{\bar{C}} \sum_{1 \leq j
 \leq \bar{C}} \sum_{i=1}^{N_j} f(W_{j,i},T_{j},D_{j,i}) - E\left[\sum_{i=1}^{N_1} f(W_{1,i},T_{1},D_{1,i}) \right] \right)
\]
and
 \begin{align*}
  \hat{G}^{C,1,d}(y_{1,d}w) =&
  \sqrt{\bar{C}}\left(\hat{F}^{C}_{Y_b|W,T,D}(y^{1,d}|w,1,d) -
  F_{Y_b|W,T,D}(y^{1,d}|w,1,d)\right), \forall y^{1,d} \in \mathcal{Y}_b^{w,\{1,d\}}\\
  \hat{G}^{C,1}(y^1,w) =& \sqrt{\bar{C}}\left(\hat{F}^{C}_{Y_b|W,T}(y^1|w,1) -
  F_{Y_b|W,T}(y^1|w,1)\right), \forall y^1 \in \mathcal{Y}_b^{w,1} \\
  \hat{G}^{C,0}(y_0,w) =& \sqrt{\bar{C}}\left(\hat{F}^{C}_{Y(0)|W,T}(y^0|w,0) -
  F_{Y(0)|W,T}(y^0|w,0)\right), \forall y^0 \in \mathcal{Y}^{w,0},
 \end{align*}
where
\[
 \hat{F}^C_{Y|W,K}(y|w,k) = \Lambda (w'\hat{\beta}^{C,k}(y))
\]
and
\begin{align*}
 \hat{\beta}^{C,k}(y) = \arg \max_{b \in \mathbb{R}^{dW}} &\sum_{j =
 1}^{\bar{C}}\sum_{i = 1}^{N_j}\left\{ \right.
 \left[I\{Y_{j,i} \leq y\}\log[\Lambda(w'b)] \right]\\
 &+ \left. \left[I\{Y_{j,i} > y\}\log[1 - \Lambda(w'b)] \right]\right\}I\{K_{j,i}
 = k\}
\end{align*}
for each $y \in \mathcal{Y}_k$.

\begin{corollary}
 Under assumptions for Theorem 3.2, and, Assumptions \ref{ass:condDGPcluster} and \ref{ass:QRFRcluster}, we have
 \[
  \left(\hat{G}^{C,1}(y,w),\hat{G}^{C,1,d}(y,w),\hat{G}^{C,0}(y,w)\right)
 \rightsquigarrow \left(\mathbb{G}^{C,1}(y,w),\mathbb{G}^{C,1,d}(y,w),\mathbb{G}^{C,0}(y,w) \right)
 \]
 in $l^{\infty}(\mathcal{Y}_b\times\mathcal{W} \times
 \mathcal{Y}_b^{1,d}\times\mathcal{W}_d \times
 \mathcal{Y}^0\times\mathcal{W})$, where $\mathbb{G}^k(y,w)$ for every
 $k \in \{\{1,d\}_{d \in \{0,1\}},1,0\}$ are tight zero-mean Gaussian processes with
 each covariance function of the form
 \begin{align*}
  V^C_{k,k}(y,w,\tilde{y},\tilde{w})&
  = \alpha_{k}^{-1} w'J_{k}^{-1}(y)\lambda_k(w'\beta^k(y)) \Sigma_k(y,\tilde{y})
  \lambda_{k}(\tilde{w}'\beta^{k}(\tilde{y}))
  J_{k}^{-1}(\tilde{y}) \tilde{w}\\
  V_{1,\{1,d\}}(y,w,\tilde{y},\tilde{w})&
  = \alpha_{1,d}^{-1} w'J_{1}^{-1}(y)\lambda_1(w'\beta^1(y)) \Sigma_{1,\{1,d\}}(y,\tilde{y})
  \lambda_{1,d}(\tilde{w}'\beta^{1,d}(\tilde{y}))
  J_{1,d}^{-1}(\tilde{y}) \tilde{w}
 \end{align*}
 where
 \begin{align*}
    \Sigma^C_k(y,\tilde{y}) =& E[\sum_{i=1}^{N(1)}I\{K_{1,i} = k\}W_{1,i}H(W_{1,i}'\beta^k(y))\\
  \times &\{\min\{\Lambda(W_{1,i}'\beta^k(y)),\Lambda(W_{1,i}'\beta^k(\tilde{y}))\} -
  \Lambda(W_{1,i}'\beta^k(y))\Lambda(W_{1,i}'\beta^k(\tilde{y}))\}\\
  & \times H(W_{1,i}'\beta^k(\tilde{y}))]W_{1,i}',\\
  \Sigma^{C}_{1,\{1,d\}}(y,\tilde{y}) =& E[\sum_{i=1}^{N(1)} I\{T_{1} = 1,D_{1,i} = d\}W_{1,i}H(W_{1,i}'\beta^1(y))\\
  \times &\{\min\{\Lambda(W_{1,i}'\beta^1(y)),\Lambda(W_{1,i}'\beta^{1,d}(\tilde{y}))\} -
  \Lambda(W_{1,i}'\beta^{1}(y))\Lambda(W_{1,i}'\beta^{1,d}(\tilde{y}))\}\\
  & \times H(W_{1,i}'\beta^{1,d}(\tilde{y}))W_{1,i}']
 \end{align*}
 for each $k \in \{\{1,d\}_{d \in \{0,1\}},1,0\}$
 and $V^C_{1,0}(y,w,\tilde{y},\tilde{w}) =
 V^C_{\{1,d\},0}(y,w,\tilde{y},\tilde{w}) = 0$, and we have
 \[
  \hat{G}^C(f) \rightsquigarrow \mathbb{G}^C(f)
 \]
 in $l^{\infty}(\mathcal{F})$ for suitably measurable $\mathcal{F}$ specified earlier where
 $\mathbb{G}^C(f)$ is a tight zero-mean Gaussian process with the
 covariance kernel
 \[
  V^C(f_1,f_2) = Cov\left(\sum_{i=1}^{N(1)}
 f_1(W_{1,i},D_{1,i},T_{1}),\sum_{i=1}^{N(1)}
 f_2(W_{1,i},D_{1,i},T_{1}) \right)
 \]
\end{corollary}
\begin{proof}
 It is sufficient to verify conditions for Theorem 3.1 in
 \cite{DaveziesDHaultfoeuilleGuyonvarch18}. Assumption 1 is guaranteed
 by the exchangeable cluster assumption. Assumption 2 is assumed and
 can be verified to be pointwise measurable. For Assumption 3 in
 \cite{DaveziesDHaultfoeuilleGuyonvarch18}, it is sufficient to show
 that the envelope function of the class of Z-functions $F$ satisfies
 \[
  E\left[N_1 \sum_{i=1}^{N_1} F(Y_{1,i},W_{1,i},D_{1,i},T_{1})^2\right] < \infty
 \]
 In fact, the envelope function is $const \cdot \|W_i\|$ and therefore it
 is sufficient to have
 \[
  E\left[N_1 \sum_{i=1}^{N_1} \|W_{1,i}\|^2\right] < \infty
 \]
 which is guaranteed by the assumption. Finally, the finiteness of the
 uniform entropy integral is also guaranteed by Lemma
 \ref{lmm:Donsker} so that all the assumptions hold.

 As the weak convergences of the Z-function processes are guaranteed,
 Functional Delta method guarantees the statement to hold.
\end{proof}

\section{Proofs for the statements in the main text}

\begin{proof}[Proof for Theorem 3.2]
 The counterfactual cdf $F_{Y(0)|W,T,D}(y|w,1,d)$ is expressed as
 \begin{align*}
  F_{Y(0)|W,T,D}(y|w,1,d) =& E[1\{Y(0) \leq y\}|W=w,T=1,D=d]\\
  =& E[1\{F_{Y(0)|W,T}(Y(0)|w,1) \leq F_{Y(0)|W,T}(y|w,1)\}|W=w,T=1,D=d]\\
  =& E[1\{F_{Y(0)|W}(Y(0)|w) \leq F_{Y(0)|W}(y|w)\}|W=w,T=1,D=d]\\
  =& E[1\{U_{0,w} \leq F_{Y(0)|W}(y|w)\}|W=w,T=1,D=d]
 \end{align*}
 in terms of the conditional latent rank $U_{0,w}$. In the second and third
 equality, I use the assumption that $Y(0) \indep T| W$ and therefore,
 the support of $Y(0)|W = w,T = 1,D = d$ is the subset of the support of
 $Y(0)|W = w, T = 1$ that equals to the support of $Y(0)|W = w$. The conditional rank
 similarity assumption implies
 \begin{align*}
  E&[1\{U_{0,w} \leq F_{Y(0)|W}(y|w)\}|W=w,T=1,D=d] \\
  =& E[1\{U_{b,w} \leq F_{Y(0)|W}(y|w)\}|W=w,T=1,D=d]\\
  =& E[1\{F_{Y_b|W}(Y_b|w) \leq F_{Y(0)|W}(y|w)\}|W=w,T=1,D=d].
 \end{align*}
From the unique quantile transformation with $Y_b$ and $Y_b \indep T| W$, we have
\begin{align*}
 E&[1\{F_{Y_b|W}(Y_b|w) \leq F_{Y(0)|W}(y|w)\}|W=w,T=1,D=d]\\
 =&E[1\{Y_b \leq Q_{Y_b|T,W}(F_{Y(0)|W}(y|w)|1,w)\}|W=w,T=1,D=d]\\
 =& F_{Y^b|W,T,D}(Q_{Y_b|T,W}(F_{Y(0)|T,W}(y|0,w)|1,w)|w,1,d).
\end{align*}
\end{proof}

\begin{proof}[Proof of Theorem 6.2]
 For the proof, I omit the covariates as the same argument goes through
 under the common support assumption. Also, I omit $T$ from most of the
 conditioning variables using the randomization assumption. From the monotonicity, we have
 \begin{align*}
  F_{Y(0)|D = 0,T = 1}&(y) = P(Y(0) \leq y|D(1) = 0) = P(Y(0,0) \leq y|D(1) = 0,D(0) = 0)\\
  =& P(U_{00} \leq F_{Y(0,0)|D = 0, T =0}(y)|D(1) = 0,D(0) = 0)
 \end{align*}
 for any value of $y$ in the support of $Y(0)|D(1) = 0$ because the
 support of $Y(0)|D(1) = 0$ is the subset of the support of $Y(0)|D(0) = 0$
 from the monotonicity. From the rank similarity, we have
 \begin{align*}
  P(U_{00} &\leq F_{Y(0,0)|D = 0,T = 0}(y)|D(1) = 0,D(0) = 0)\\
  =& P(U_{b0} \leq F_{Y(0,0)|D = 0, T = 0}(y)|D(1) = 0,D(0) = 0).
 \end{align*}
 The continuity of the baseline varible allows us to invert back in
 the level
 \begin{align*}
  P(U_{b0} &\leq F_{Y(0,0)|D = 0,T = 0}(y)|D(1) = 0,D(0) = 0)\\
   =& P(Y_b \leq Q_{Y_b|D=0,T=0}(F_{Y(0,0)|D = 0,T = 0}(y))|D(1) = 0,D(0) = 0)\\
  =& P(Y_b \leq Q_{Y_b|D=0,T = 0}(F_{Y(0,0)|D = 0,T = 0}(y))|D(1) = 0),
\end{align*}
 and the last equality follows from the monotonicity again.
\end{proof}

\section{Additional discussions}

\subsection{Stability of the treatment and the rank similarity assumption}
In Section 6, I show the identification with \textit{stable} two-sided noncompliance cases. In Corollary 6.1, I use the result below which shows that the conditional rank similarity for each control outcome $Y(0,1)$ and $Y(0,0)$, and the \textit{stable} take-up behavior give the conditional rank similarity for $Y(0)$:
\begin{lemma} 
Suppose Assumptions 3.1, 3.2, 3.3, 3.5, 6.1 and 
\[
  P(D(0) = 1|D(1) = d, W = w) = P(D_b = 1|D(1) = d, W = w)
 \]
 holds for each $d \in \{0,1\}$ and $w \in \mathcal{W}$, where $\mathcal{W}$ represents the intersection of the conditional supports of $W$ given $\{D = d,T = t\}$ for each $t \in \{0,1\}$ and $d \in \{0,1\}$. Then, we have
\[
 U_{0,w} \sim U_{b,w} | T = 1, D = d, W = w
\]
for each $d \in \{0,1\}$ and $w \in \mathcal{W}$.
\end{lemma}
\begin{proof}
In the proof, I omit $w$ as it does not alter the logic of the proof. Note that we have
 \[
  U_{0} = U_{0,1} D(0) + U_{0,0} (1 - D(0))
 \]
 and
 \[
  U_{b} = U_{b,1} D_b + U_{b,0} (1 - D_b).
 \]
This is because the construction of the latent rankings is linear in their distribution functions, and $V$ in the contruction is independent of $Y$. Therefore,
\begin{align*}
 P&(U_0 \leq u|D(1) = d) = P(U_{0,1} D(0) + U_{0,0}(1 - D(0)) \leq u|D(1) = d) \\
 =& P(U_{0,1} \leq u| D(1) = d,D(0) = 1)P(D(0) = 1|D(1) = d)\\
&+ P(U_{0,0} \leq u| D(1) = d,D(0) = 0)P(D(0) = 0|D(1) = d)\\
 =& P(U_{b,1} \leq u| D(1) = d,D_b = 1)P(D_b = 1|D(1) = d)\\
&+ P(U_{b,0} \leq u| D(1) = d,D_b = 0)P(D_b = 0|D(1) = d)\\
=& P(U_b \leq u|D(1) = d)
\end{align*}
which concludes the statement.
\end{proof}

\subsection{Identification of the net effects in two-sided noncompliance designs}
In Section 6, I illustrate the identification of the net effects of the treatment $D$ in two-sided noncompliance:
 \begin{equation}\tag{4}
  E[Y(0,1) - Y(0,0)|T = 1, D(1) > D(0)].
 \end{equation}

Let me introduce additional notations
\[
 DE(1) \equiv Y(1,1) - Y(0,1),  DE(0) = Y(1,0) - Y(0,0)
\]
which are direct effects with and without the treatment $D$.

\begin{corollary}
 Suppose we identify the ITTTA 
 \begin{equation}\tag{1}
  E[Y(1) - Y(0)|T = 1, D = 1],
 \end{equation}
 and the ITTNA
 \begin{equation}\tag{2}
  E[Y(1) - Y(0)|T = 1, D = 0].
 \end{equation}
 Assume $D(1) \geq D(0)$ almost surely and $P(D=1|T=1) > 0$.

 If $E[DE(1)|D(1)= 1] = E[DE(0)|D(1)=0]$ then we identify $E[Y(0,1) - Y(0,0)|D(1) > D(0)]$. 
\end{corollary}
\begin{proof}
Note first that $ITTNA = E[Y(1) - Y(0)|T=1,D(1) = 0] = E[Y(1,0) - Y(0,0)|D(1) = 0] = E[DE(0)|D(1) = 0]$. 
\begin{align*}
 ITTTA =& E[Y(1) - Y(0)|D(1) = 1] = E[Y(1,1) - Y(0,D(0))|D(1)=1]\\
=& E[Y(1,1) - Y(0,1) + Y(0,1) - D(0)Y(0,1) - (1 - D(0))Y(0,0)|D(1)=1]\\
=& E[DE(1) + (1 - D(0))(Y(0,1) - Y(0,0)|D(1)=1]\\
=& E[(1 - D(0))(Y(0,1) - Y(0,0))|D(1) = 1] + E[DE(1)|D(1)=1]\\
=& E[Y(0,1) - Y(0,0)|D(1) > D(0)]\frac{P(D(1) > D(0))}{E[D(1)]} + E[DE(1)|D(1)=1].
\end{align*}
Therefore,
\[
 E[Y(0,1) - Y(0,0)|D(1) > D(0)] = (ITTTA - ITTNA) \frac{P(D(1) > D(0))}{P(D(1) = 1)}.
\]
\end{proof}

\begin{remark}

This homogeneity assumption is stronger than that of one-sided noncompliance in Section 3 and more related to additivity or the lack of interaction effects, as in \cite{Holland_1988} and \cite{Robins_2003} (see \citealp{Imai_Tingley_Yamamoto_2013}). While homogeneity in one-sided noncompliance only limits the heterogeneity of the direct effect, we must also limit the interaction effect of $T$'s direct effect with $D$'s treatment effect. In this context, the existence of the interaction effect is defined as
\[
 Y_i(1,1) - Y_i(1,0) \neq Y_i(0,1) - Y_i(0,0).
\]
Identifying detailed causal effect mechanisms without imposing the assumption of no-interaction effect is challenging. For example, \cite{Imai_Tingley_Yamamoto_2013} proposed a specific experimental design, called a crossover design, to identify the average causal indirect effect, $E[Y(1,D(1))-Y(1,D(0))]$ and $E[Y(0,D(1))-Y(0,D(0))]$. With these indirect effects, we may calculate the interaction effect for compliers
\begin{align*}
 E&[Y(1,D(1)) - Y(1,D(0))] - E[Y(0,D(1)) - Y(0,D(0))] \\
 &= E[Y(1,1) - Y(1,0) - (Y(0,1) - Y(0,0))|D(1) > D(0)] 
\end{align*}
where under monotonicity $D(1) \geq D(0)$ almost surely. 
\end{remark}

Although the homogeneity assumption is particularly strong for two-sided noncompliance, the homogeneity restriction is substantially weaker than that in the no-direct effect assumption, which postulates a constant direct effect for which the constant equals 0. Homogeneity may be a reasonable restriction if assignment $T$ has an additive impact on the outcome when the treatment status $D$ is left unchanged. For example, assigning a fixed amount of $T$ transfers may have an additive impact on consumption when investment level $D$ is fixed. Appendix F offers further discussion in the context of an empirical application.


\subsection{Attrition as a special case of two-sided noncompliance}
An important example of two-sided noncompliance, when the endogenous \textit{treatment} is not available in the baseline, is the selective attrition problem, a serious concern in experimental evaluations. Below, I demonstrate that the ITTNA alone is sufficient to uncover the most relevant parameter in the context of selective attrition. In an evaluation problem for assignment $T$, in which $D$ is an attrition indicator in the follow-up survey for the outcome variable $Y$,
\[
 D = \begin{cases}
      1 &\mbox{ if } Y \mbox{ is missing in the follow-up survey},\\
      0 &\mbox{ if else}.
     \end{cases}
\]

By construction, the attrition event is realized only after the $T$ assignment. Regarding interventions in unemployment benefit systems, \citeauthor{vandenBergVikstrom2014}'s (\citeyear{vandenBergVikstrom2014}) Swedish study examined the extent to which punishments for noncompliance to job search requirements affect job quality. They reported that wages, serving as the job quality measure, contain nonrandom missing values because their sample covered mostly large firms, with only a fraction of smaller firms. Thus, attrition D may be associated with punishment intervention $T$ because the unemployed may have led workers to accept poor quality jobs in response to the punishment. In other words, there may be an issue of differential attrition, that is, $D(1) \neq D(0)$.
Furthermore, if we consider that potential outcomes $Y(1)$ and $Y(0)$ are truncated observations of the true potential outcomes $Y^*(1), Y^*(0)$, then
\[
 Y(1) = (1 - D(1)) Y^*(1), Y(0) = (1 - D(0)) Y^*(0).
\]
$Y^*(1), Y^*(0) > 0$ to ensure that 0 is only taken by the missing values. This definition is a normalization available for the bounded supported $Y^*(t)$ by adding the minimum of $Y^*(t)$ to make it positive. Additionally, the monotonicity, $D(1) \geq D(0)$ almost surely, may be reasonable in the earlier example, as attrition is more likely to occur with punishment intervention $T$ because such a punishment may have pushed unemployed workers to join smaller firms that are not covered by the dataset. Based on the above specification of the observed potential outcomes $Y(1)$ and $Y(0)$, identifying the counterfactual distribution
\[
 F_{Y(0)|T = 1,D(1) = 0}(y) = F_{Y(0)|T=1,D(1) = 0,D(0) = 0}(y),
\]
by Theorem 6.2, relates its underlying potential outcome distribution because
\[
 F_{Y(0)|T=1,D(1) = 0,D(0) = 0}(y) = F_{Y^{*}(0)|T=1,D(1) = 0,D(0) = 0}(y).
\]
Therefore, we can identify intervention $T$'s causal effect for those who always took the survey, $D(1) = D(0) = 0$, as
\begin{align*}
 E[Y(1) - Y(0)|T = 1, D(1) = 0] =& E[Y(1) - Y(0)|T = 1, D(1) = 0,D(0) = 0]\\
 =& E[Y^*(1) - Y^*(0)|D(1) = 0, D(0) = 0].
\end{align*}

\subsection{A set of conditions sufficient for rank similarity}

In Section 3, I illustrate the identification process with the rank similarity's reduced-form restriction on the potential outcomes. The target readers may wonder how restrictive the rank similarity assumption is and what kinds of covariates $W$ may help ease the restriction. Here, I offer a set of low-level conditions for the rank similarity. Remember that we require the following models for illustrative purposes only, to consider candidates for $Y_b$ and $W$.

To fix the ideas, I consider \citeauthor{CreponDevotoDufloPariente2015}'s (\citeyear{CreponDevotoDufloPariente2015}) microcredit experiment for which $Y$ is the production output sales value two years after the experiment and $Y_b$ is the same sales value measured before the experiment. Although I demonstrate the simplest case of the repeated outcome measures in the example and application, the same argument applies to the binary indicator of business performance as the outcome of interest.

To rationalize the rank similarity between $Y_b$ and $Y(0)$, I consider
scalar latent productivity $U_{b,w}$ and $U_{0,w}$ that determine the rankings of variables $Y_b$ and $Y(0)$ conditional on other baseline predetermined variables $W$. Such scalar rankings stem from Definition 3.1, but it sometimes helps to consider the rankings $U_{b,w}$ and $U_{0,w}$ as having the same meanings (e.g., productivity). In other words, people with higher $U_{b,w}$ and $U_{0,w}$ are likely to have higher latent productivity. The conditioning variables $W$ help rationalize such a monotone relationship, as $W$ may include current borrowing statuses or household characteristics, other than productivity.


To further explore the restriction, I consider that there is a common latent variable $U$ such that $U_{b,w} = U$ and $U_{0,w} = U$ almost surely; then, the rank similarity trivially holds. While this specification is an extreme example, called the rank invariance restriction, this is a fair starting point to relax the restrictions. In particular, the $U_{b,w}$ and $U_{0,w}$ rankings may accept random permutations, called \textit{slippages}, from the common level $U$ (\citealp{HeckmanSmithClements97}). For example, if $\tilde{U}_{b,w}$ and $\tilde{U}_{0,w}$ are a pair of identically distributed random variables and normalization is achieved through the common distribution function $F$, and I can define latent productivity in the following forms:
\begin{equation}\tag{A.1}
 U_{b,w} \equiv F(U + \tilde{U}_{b,w}), \; U_{0,w} \equiv F(U + \tilde{U}_{0,w}). \label{eq:Slip}
\end{equation}
This specification is substantially weaker than that of rank invariance because it enables an arbitrary positive correlation between 
$U_{0,w}$ and $U_{b,w}$.

If $V$ represents the unobserved determinants $V$ of taking up treatment $D$ and the slippages are independent of $V$, but conditional on their potential productivity $U$, then the rank similarity assumption holds.
\begin{proposition} \label{prop:suffRS}
 Let $U_{b,w}$ and $U_{0,w}$ be the latent rank variables for $Y_b$
 and $Y(0)$ conditional on $W$ and constructed as in (\ref{eq:Slip}). Assume
 that there is a vector of unobservables $V$ such that $D =
 \delta(U,V,T,W)$ for some measurable map $\delta$, and $(\tilde{U}_{b,w},\tilde{U}_{0,w}) \indep
 V|U,T=1,W$ and $\tilde{U}_{b,w} \sim \tilde{U}_{0,w}|U,T=1,W$. Then Assumption 3.4 holds.
\end{proposition}
\begin{proof}
  For the simplify of notations, omit $W$ in the expressions.
 \begin{align*}
  P&(F(U_{b}) \leq \tau|D = d,T = 1) = P(F(U + \tilde{U}_{b}) \leq \tau|\delta(U,V,1) = d,T = 1)\\
  =& \int P(F(u + \tilde{U}_{b}) \leq \tau|\delta(u,V,1) = d,T = 1,U=u) dF_{U|\delta(u,V,1) = d,T = 1}(u)\\
  =& \int P(F(u + \tilde{U}_{b}) \leq \tau|T = 1,U=u) dF_{U|\delta(u,V,1) = d,T = 1,U=u}(u)\\
  =& \int P(F(u + \tilde{U}_{0}) \leq \tau|T = 1,U=u) dF_{U|\delta(u,V,1) = d,T = 1,U=u}(u)\\
  =& \int P(F(u + \tilde{U}_{0}) \leq \tau|\delta(u,V,1) = d,T = 1,U=u) dF_{U|\delta(u,V,1) = d,T = 1,U=u}(u)\\
  =& P(F(U_{0}) \leq \tau|D = d,T = 1)
 \end{align*}
\end{proof}
The second assumption claims that the marginal distributions of slippages are the same, but conditional on the same potential productivity $U$. However, the joint distribution of slippages is not restricted to ensure that they can correlate with each other. The above proposition may help us rationalize the rank similarity through a model for taking up a treatment.
\begin{example}
Consider a threshold crossing model of take-up decision $D$ for repaying microcredit debt. Assume that the business owners decide whether to borrow from a microcredit institution by comparing the expected revenue with loan-enabled investments. $U$ represents the private information of business owners; they can make a subjective prediction of performance when they take up the treatment: $\tilde{E}[Y(1,1)|U]$. Under this, $\tilde{Y}$ is the random repayment amount, which is expected to increase when the other microcredit group members fail to repay. As a result, the decision to take up the treatment may be written as
\[
 D = 1\{\tilde{E}[Y(1,1)|U] \geq \tilde{Y}\}.
\]
If $\tilde{Y}$ is independent of $Y_b$ given $U$, then the former condition for the above proposition holds. 
In other words, if we fix the true profitability of the business $U = u$, then today's bad luck, $\tilde{U}_b$, should be independent of the expected repayment amount $\tilde{Y}$. There may be other reasons why $\tilde{Y}$ may change systematically. For instance, today's bad luck, $\tilde{U}_b$, may have affected the baseline loan status which, in turn, may have impacted the borrowing amount $\tilde{Y}$. To prevent such a correlation, it is important to include the baseline credit status and other household characteristics as the conditioning covariates $W$.
\end{example}





\subsection{Rank similarity justification for a peer effect model}
Here, I consider how peer or equilibrium effects, due to exposure $T = 1$, are compatible with the identification restriction. Consider that each person $i$ faces his/her community $c$. Then, this would affect his/her outcome as an equilibrium response to the common assignment $T_c = 1$ within $c$. For each person $i$ in community $c$, the underlying potential outcomes $Y_{i;c}(1,1)$ and $Y_{i;c}(1,0)$ would be explained as
 \[
  Y_{i;c}(1) = D_{i;c} Y_{i;c}(1,1) + (1 - D_{i;c}) Y_{i;c}(1,0)
 \]
 Under this, Assumption 3.1 remains valid. For example, if we consider $D_{i;c}$ with a generalized Roy model, then there are $Y_{i;c}(1,d), d \in \{0,1\}$ defined with $D_{i;c}
 = 1\{Y_{i;c}(1,1) \geq Y_{i;c}(1,0) + \epsilon_i\}$ for each $i$ in community $c$.
I perceive the strategic interactions such that the individual threshold crossing may not solely determine taking up treatment $D$. Rather, these potential outcomes depend on neighbors taking up the treatment, $D^{-i}_{c}(1,D_{i;c})$ and their decisions may depend on $i$'s behavior $D_{i;c}$. Assume that the same potential productivity $U_i$, as for the control, and slippages $\tilde{U}_{i}(1,1)$ and $\tilde{U}_{i}(1,0)$ determines the potential outcomes with assignment $T_c = 1$. Then, $Y_{i;c}(1,1) = y_{1,1}(U_i + \tilde{U}_{i,11};D^{-i}_{c}(1,1))$
 and $Y_{i;c}(1,0) = y_{1,0}(U_i + \tilde{U}_{i,10};D^{-i}_{c}(1,0))$ for some functions $y_{1,1}$ and $y_{1,0}$. In other words, the Roy model becomes
 \[
  D_{i;c} = 1\{y_{1,1}(U_i + \tilde{U}_{i,11};D^{-i}_{c}(1,1)) \geq
 y_{1,0}(U_i + \tilde{U}_{i,10};D^{-i}_{c}(1,0)) + \epsilon_i\}
 \]

Here, I assume that Assumption 3.1 guarantees the existence of a unique equilibrium vector for taking up treatment $D_{i;c}$ as the solution to the above interaction model.

In this case, Proposition \ref{prop:suffRS} justifies the rank similarity for $U_b$ and $U_0$ if
 \[
  (\{U_j,\tilde{U}_{j,11},\tilde{U}_{j,10},D_{j;c}\}_{j \in
 c\backslash i}, \epsilon_i) \indep (\tilde{U}_{i,0},\tilde{U}_{i,b}) |
 U_i
 \]
in addition to $\tilde{U}_{b,w} \sim \tilde{U}_{0}|U,T=1$, which represents the assumption solely on the control outcomes. As the above condition does not impose independence of $U_i$ and $U_j$ within the community, the outcomes may correlate with each another.

\subsection{Illustrating a binary outcome variable compatible with the identification assumptions}
In Section 3.1, I noted that Assumption 3.3 and 3.4 do not rule out a discrete $Y(0)$ with a continuous $Y_b$. Below I offer an example with a binary $Y(0)$ satisfying both restrictions.

For simplicity, ignore the covariates, but note that they allow us to make the model much more flexible. Consider that $Y_b$ is a uniform random variable,
        \[
         F_{Y_b}(y) = y, y \in [0,1].
        \]
        This is strictly increasing and continuous over $[0,1]$, satisfying Assumption 3.3. With this distribution, we can define the latent variable for $Y_b$ as
        \[
         U_b \equiv Y_b \sim U[0,1].
        \]

        Conversely, if we consider $V_0 \sim U[0,1]$ and
        \[
         Y(0) = 1\{V_0 \geq 0.8 \},
        \]
        the distribution function is neither continuous nor strictly increasing over $[0,1]$ as
        \[
         0 = F_{Y(0)}(0-) < F_{Y(0)}(0) = F_{Y(0)}(1-) = 0.8 < F_{Y(0)}(1) = 1,
        \]
        but $V_0$ works as the latent rank variable for $Y(0)$, $U_0$. To perceive this, $V \sim U[0,1]$ such that
        \begin{align*}
         V_0 =& 0.8 Y(0) + V\left[0.2 Y(0)  + 0.8 (1 - Y(0))\right] = 0.8 Y(0) + V (0.8 - 0.6 Y(0)).
        \end{align*}
        Then,
        \[
         V = \frac{V_0 - 0.8 Y(0)}{0.8 - 0.6 Y(0)}.
        \]
        This $V$ is consistent with Definition 3.1 because for $0 \leq v \leq 1$,
        \[
         P(V \leq v|Y(0) = 0) = P\left(\frac{V_0}{0.8} \leq v \middle| V_0 < 0.8\right) = v
        \]
        as $V_0|V_0 < 0.8 \sim U[0,0.8]$. Similarly for $0 \leq v \leq 1$,
        \[
         P(V \leq v|Y(0) = 1) = P\left(\frac{V_0 - 0.8}{0.2} \leq v \middle| V_0 \geq 0.8\right) = v
        \]
        as $(V_0 - 0.8)|V_0 \geq 0.8 \sim U[0,0.2]$. Thus, $V \indep Y(0)$.
        
        With these constructions, Assumption 3.4 imposes a restriction on $U_0$ and $U_b$, or equivalently $V_0$ and $Y_b$:
        \[
         Y_b \sim V_0 |T = 1, D = d.
        \]
        Therefore, discreteness of $Y(0)$ induced by the map
        \[
         Q_{Y(0)}(u) = 1\{u \geq 0.8\}
        \]
        does not contradict with Assumptions 3.3 and 3.4.

\newpage
\section{Additional Plots and Robustness Checks}
\subsection{Quantile Difference Plots}

  \begin{figure}[h!]
   \centering
  \includegraphics[height=0.8\textwidth,width=0.8\textwidth]{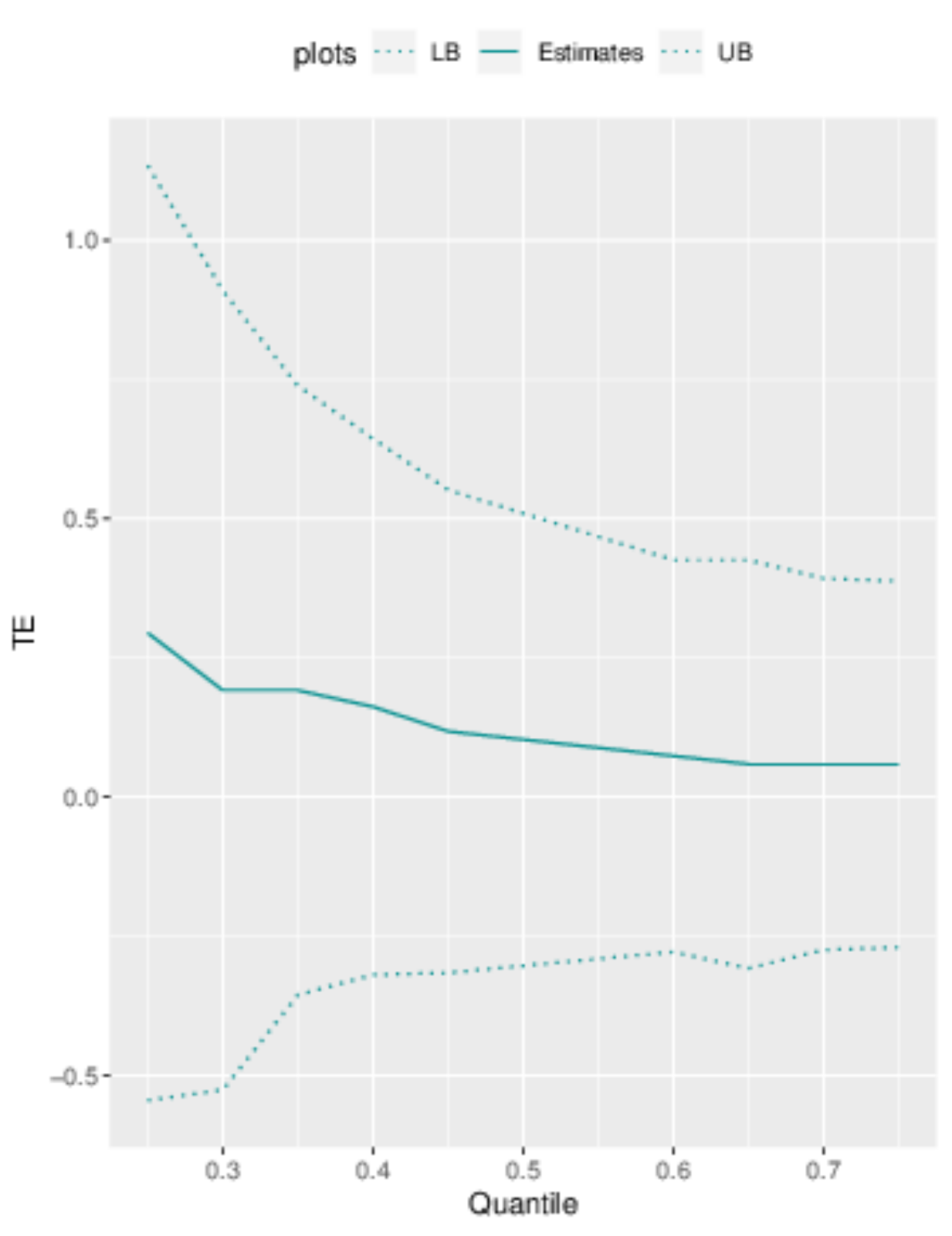}
 \caption{Quantile difference of the direct effects for the nontreated, D = 0, with Uniform 95\% CI. The line
   in center indicates the quantile difference estimates at each value
   of quantiles. The dotted lines indicate the 95\% CI of the estimates
   uniformly valid for the specified range.}
  \label{fig:MFQTE0}
 \end{figure}

\newpage

  \begin{figure}[h!]
   \centering
  \includegraphics[height=0.8\textwidth,width=0.8\textwidth]{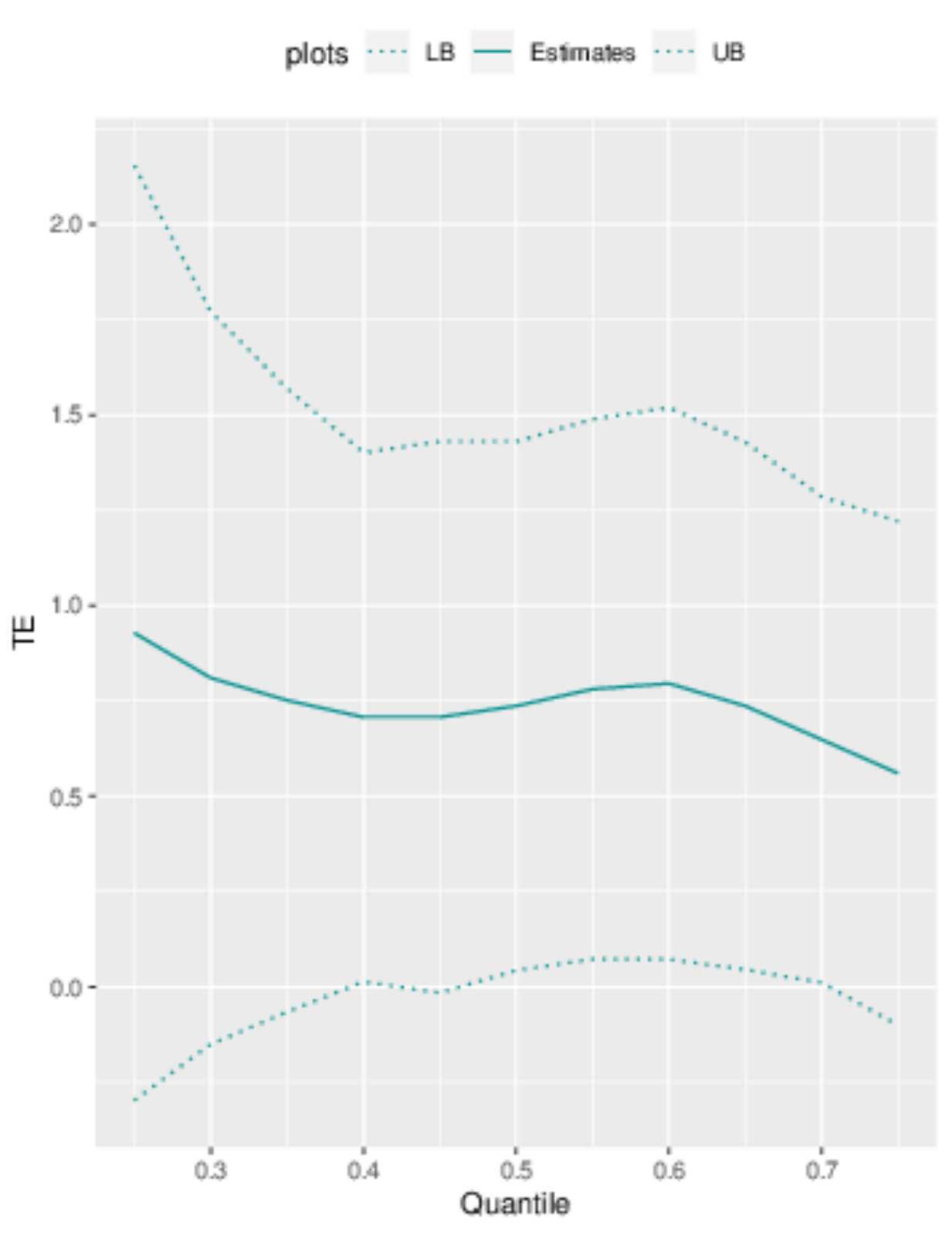}
 \caption{Quantile difference of the combined effects for the treated, D = 1, with Uniform 95\% CI. The line
   in center indicates the quantile difference estimates at each value
   of quantiles. The dotted lines indicate the 95\% CI of the estimates
   uniformly valid for the specified range.}
  \label{fig:MFQTE1}
 \end{figure}

\newpage

\subsection{Robustness to other link functions}
\begin{table}[H]
 \centering
  \caption{ITTTA estimates with other link functions} 

 \begin{tabular}{lcccc}
 \hline
 & (1) & (2) & (3) & (4)\\
& & & complementary & \\
Link & Logit & Probit & Log-Log & Cauchy\\
Parameter & ITTTA & ITTTA & ITTTA & ITTTA \\ \cline{2-5}
Outcome & \multicolumn{4}{c}{log output}\\
\hline
  Assignment ($T$) & $1.0252^{**}$ & $1.0280^{**}$ & $1.0286^{**}$ & $1.0548^{**}$\\
  by subgroups of ($D$)          & ($0.4492$) & ($0.4030$) & ($0.4304$)& ($0.4247$)\\
 \hline
  Self-employed in baseline & Y & Y & Y & Y\\
  Obs       & $2,453$ & $2,453$  & $2,453$ & $2,453$\\
  \hline
 \end{tabular}
\begin{minipage}{325pt}
{\flushleft \fontsize{9pt}{9pt}\selectfont \smallskip NOTE: Standard
 errors reported in parenthesis are generated from 300 bootstrap draws clustered in
 randomization cluster levels for (1)-(4). *,**,*** indicates statistical significance of
 10\%,5\% and 1\% sizes respectively.}
\end{minipage}
\end{table} 

\begin{table}[H]
 \centering
  \caption{ITTNA estimates with other link functions} 

 \begin{tabular}{lcccc}
 \hline
 & (1) & (2) & (3) & (4)\\
& & & complementary & \\
Link & Logit & Probit & Log-Log & Cauchy\\
Parameter & ITTNA & ITTNA & ITTNA & ITTNA \\ \cline{2-5}
Outcome & \multicolumn{4}{c}{log output}\\
\hline
Assignment ($T$) & $0.2837$ & $0.2805$ & $0.2867$ & $0.3192$\\
by subgroups of ($D$)            & ($0.2842$) & ($0.2828$) & ($0.2688$)& ($0.2650$)\\
 \hline
  Self-employed in baseline & Y & Y & Y & Y\\
  Obs       & $2,453$ & $2,453$  & $2,453$ & $2,453$\\
  \hline
 \end{tabular}
\begin{minipage}{325pt}
{\flushleft \fontsize{9pt}{9pt}\selectfont \smallskip  NOTE: Standard
 errors reported in parenthesis are generated from 300 bootstrap draws clustered in
 randomization cluster levels for (1)-(4). *,**,*** indicates statistical significance of
 10\%,5\% and 1\% sizes respectively.}
\end{minipage}
\end{table} 

\newpage
\subsection{Estimation of main estimates with level of outputs}
\begin{table}[H]
 \centering
  \caption{Estimates with baseline proxy $Y_b$} 

 \label{tbl:mainResultsLevel}
 \begin{tabular}{lcccc}
 \hline
 & (1) & (2) & (3) & (4)\\
Parameter & ITTNA & ITTTA & ITTNA & ITTTA \\ \cline{2-5}
Outcome & \multicolumn{4}{c}{level output}\\
\hline 
Assignment ($T$)   & $4,081$ & $26,664^{**}$ & $5,161$ & $28,567^{**}$\\
by subgroups of ($D$)      & ($5,016$) & ($10,883$) & ($5,628$) & ($12,222$)\\ 
 \hline
  Self-employed in baseline & Y & Y & Y & Y\\
  Trimming at 99 percentile &  & & Y & Y\\
  Obs       & $2,453$ & $2,453$ & $2,429$ & $2,429$\\ 
  \hline
 \end{tabular}
 \vspace{0.3cm}
\begin{minipage}{325pt}
{\flushleft
 \fontsize{9pt}{9pt}\selectfont \smallskip NOTE: Standard
 errors reported in parenthesis are generated from 100 bootstrap draws clustered in
 village levels for (1)-(3). *,**,*** indicates statistical significance of
 10\%,5\% and 1\% sizes respectively. Logit link is used for (1) and
 (2).} 
\end{minipage}
\end{table} 
\newpage 

\subsection{Diagnostic test for rank similarity}
\citet{DongShen18} and \citet{FrandsenLefgren18} proposed diagnostic tests for the rank similarity assumption in the same form as \citet{ChernozhukovHansen05}. However, their testable restrictions were meant for a rank similarity restriction on the unconditional latent rank rather than the conditional latent rank. As rejecting their tests does not violate the conditional latent rank assumption, we may still validate the conditional rank similarity by conditioning on such a harmful covariate. Nevertheless, I implement a version of \citeauthor{DongShen18}'s (\citeyear{DongShen18}) means test, a chi-squared test of mean differences in the latent ranks, for the unconditional latent ranks. In my application, I compare the conditional mean of the latent ranks for $Y_b$ and $Y(0)$. However, my theoretical flexibility allows $Y(0)$ to have a point mass. Therefore, the raw latent rank cannot be observed for a subset of units with $Y(0) = 0$. Instead, I compute the latent rank of $Y_b$ and $Y(0)$ among those that had $Y(0) > 0$ to ensure that their latent ranks are continuous in both variables.

Following the analogous argument of \citeauthor{DongShen18}'s (\citeyear{DongShen18}) means test, I take $J$ different values for the covariates $W$, denoted by $w_1,\ldots, w_J$, and test the null hypothesis of 
\[
 E[U_0|W = w_j] = E[U_b|W = w_j]
\]
being true for all $j$ against the alternative hypothesis of the null being not true. I test this null in the control group satisfying both $Y_b > 0$ and $Y(0) > 0$ with this study's sample selection. Additionally, I consider $U_0$ and $U_b$ to be the latent ranks of $Y(0)$ and $Y_b$ within the selected sample. I consider the following statistics
\[
 m^{w_j}_0 = \sum_{i:W_i = w_j} \hat{U}_{0i}
\]
and
\[
 m^{w_j}_b = \sum_{i:W_i = w_j} \hat{U}_{bi},
\]
and $m^w_0$ and $m^w_b$ are the vector of these sample means.

For the vector of the key covariates, I use the following baseline variables: number of adults residents and age of heads of households as well as the binary responses of livestock ownership, business activity, and borrowing. I split the sample into 12 subgroups based on the four binary responses and three categories each for the two continuous responses. For number of adults, I separate the sample per member groups: fewer than or equal to two members, more than two and fewer than or equal to four, and more than four. For age, I split the sample into less than 42, more than 41 and less than 55, and more than 54. Following \citet{DongShen18}, I construct the cluster bootstrapped variance matrix $\sqrt{n}(m^{w}_0 - m^{w}_b)$ denoted by V and compute the test statistic:
\[
 n(m^{w}_0 - m^{w}_b)' V^{-1}(m^{w}_0 - m^{w}_b).
\]
Using \citeauthor{DongShen18}'s (\citeyear{DongShen18}) results, the test statistics must follow the $\chi^2$ distribution with 11 degrees of freedom. The test statistics has a value of $0.9759$ and the corresponding p-value is $0.99996$. Therefore, I do not reject this null of equal conditional means of the latent ranks.

\section{Monte Carlo simulations for the proposed estimators}

To demonstrate the finite sample properties of the proposed estimators, I run two types of Monte Carlo simulations. First, I verify that the coverage of the proposed estimator is satisfactory when the number of clusters is moderate. Second, I check the validity of the estimation when \citeauthor{AtheyImbens06}' (\citeyear{AtheyImbens06}) parametric specification is not satisfied.

\subsection{Finite sample performance of the proposed estimators with clustered data}
I simulate a dataset of clustered samples drawn from the following data generating process: there are three latent ranks $U_1, U_0$ and $U_b$ for each individual $i$ in cluster $c$,
\begin{align*}
 U_{1,i;c} =& 0.8 U_{i} + 0.1 \tilde{U}_{1,i} + 0.1 \nu_c\\
 U_{0,i;c} =& 0.8 U_{i} + 0.1 \tilde{U}_{0,i} + 0.1 \nu_c\\
 U_{b,i;c} =& 0.8 U_{i} + 0.1 \tilde{U}_{b,i} + 0.1 \nu_c
\end{align*}
where $U_i \sim U[0,1]$ is an individual-specific unobserved determinant common across potential outcomes and $\tilde{U}_{d,i} \sim U[0,1]$ represents the slippages from the common latent rank with additional cluster specific heterogeneity $\nu_c \sim U[0,1]$. All these random errors are jointly independent. I define treatment $D(1)$ with assignment $T = 1$ to ensure that it depends on the common latent rank $U$:
\[
 D_i(1) = 1\{\phi U_i + (1 - \phi) \eta_i \leq 0.5\}
\]
where $\eta_i \indep U_i, \eta_i \sim U[0,1]$. I consider one-sided noncompliance by setting $D(0) = 0$ for everyone. This parameter $\phi \in [0,1]$ captures the strength of the endogeneity for taking up the treatment. I simplify the arguments by letting the predetermined covariates be independent of the latent ranks: 
\begin{align*}
 W_{c1} & \sim U[0,1]\\
 W_{c2} & \sim U[0,1]\\
 W_{d1} & = 1\{\xi \geq 0.5\}
\end{align*}
where $\xi \sim U[0,1]$. I set $(W_{c1}, W_{c2}, W_{d1})$ to be jointly independent of each other and from $(U_1,U_0,U_b)$. I formulate
\[
 Y(0) = 0.5 W_{c1} + 0.25 W_{c2} + 0.5 W_{c1} W_{c2} + 0.75 W_{c1}^2 + 0.25 W_{d1} + 0.5 W_{d1} W_{c2} + \Phi^{-1}(U_0)
\]
and
\[
 Y_b = 0.5 W_{c1} - 0.5 W_{c2} + 0.5 W_{c1} W_{c2} - W_{d1} W_{c1} W_{c2} + \Phi^{-1}(U_b)*1.5
\]
so that they have similar polynomial models, but not necessarily the same type of random variables, where $\Phi^{-1}(\cdot)$ represents the inverse of the standard normal distribution. I define $Y(1,1)$ and $Y(1,0)$ by specifying the direct and treatment effects to be added to $Y(0)$, and I allow these effects to depend on $U_1$. Specifically I define
\[
 Y(1,0) = Y(0) + 0.05 W_{c1} + 0.05 W_{c2} + 0.2 W_{c1} W_{c2} W_{d1} + \Phi^{-1}(U_1)
\]
and
\[
 Y(1,1) = Y(0) + 0.05 W_{c1} + 0.1 W_{c2} + 0.4 W_{c1} W_{c2} + 0.4 W_{d1} + 0.2 W_{c1} W_{c2} W_{d1} + \Phi^{-1}(U_1).
\]

Given this data-generating process, I consider two numbers of clusters $\{75, 150\}$ with a fixed number of observations in a cluster, $10$. For each sample size, I consider three cases of different strengths in endogeneity, $\phi \in \{0.113,0.203, 0.338\}$. The specification of taking up the treatment implies that U is negatively selected. In other words, lower $U$ units are more likely to take up the treatment. Without endogeneity, $\phi = 0$, both the treatment effects and the direct effects are positive. I choose three values of $\phi$ to ensure that the ITTTA and ITTNA equal the values in Table 4 in a large sample. Table 5 presents the simulation results.
\begin{table}[H]
 \centering
\caption{Large sample values of the ITTTA and ITTNA}
 \begin{tabular}{lccc}
\hline
 Strength of endogeneity & mild & medium & strong\\
 $\phi$ & $0.113$ & $0.203$ & $0.338$ \\
  \hline  
 ITTTA & $0.35$ & $0.3$ & $0.2$\\
 ITTNA & $0.25$ & $0.3$ & $0.4$\\
\hline
\end{tabular}
\end{table}

\begin{table}[H]
 \centering
 \caption{Estimator's finite sample performance with $75$ clusters}

 \begin{tabular}{lcccccccc}
\hline
& (1) & (2) & (3) & (4) & (5) & & \\
Parameer  & ITTTA & ITTNA & QT & QN &  IV &  & \\ \cline{2-6}
Statistic & \multicolumn{5}{c}{95\% coverage probability} & & \\
\hline
$\phi = 0.113$ & 0.952 & 0.972 & 0.976 & 0.980 & 0.092 &  &  \\
$\phi = 0.203$ & 0.946 & 0.962 & 0.974 & 0.978 & 0.048 &  &  \\
$\phi = 0.338$ & 0.946 & 0.970 & 0.958 & 0.932 & 0.016 &  &  \\
\hline
& (1) & (2) & (3) & (4) & (5) & (6) & (7) \\
Parameter & ITTTA & ITTNA & sum of QTs & sum of QNs &  IV & OLST & OLSN \\ \cline{2-8}
Statistic & \multicolumn{7}{c}{Mean absolute bias} \\
\hline
$\phi = 0.113$ & 0.076 & 0.073 & 0.937 & 0.813 & 0.262 & 0.341 & 0.091 \\
$\phi = 0.203$ & 0.075 & 0.075 & 0.832 & 0.908 & 0.307 & 0.492 & 0.118 \\
$\phi = 0.338$ & 0.076 & 0.074 & 0.761 & 1.422 & 0.405 & 0.793 & 0.199 \\ \cline{2-8}
Statistic & \multicolumn{7}{c}{Mean squared error} \\
\hline
$\phi = 0.113$ & 0.009 & 0.009 & 0.610 & 0.342 & 0.091 & 0.136 & 0.013 \\
$\phi = 0.203$ & 0.009 & 0.009 & 0.313 & 0.627 & 0.119 & 0.263 & 0.020 \\
$\phi = 0.338$ & 0.009 & 0.009 & 0.141 & 2.212 & 0.191 & 0.649 & 0.049 \\
\hline
 \end{tabular}
\begin{minipage}{325pt}
{\flushleft \fontsize{9pt}{9pt}\selectfont \smallskip NOTE: Every sample is clustered data: $75$ clusters with $10$ observations each. I computed 95\% coverage, mean absolute bias (ABias), and mean squared errors (MSE) using Monte Carlo simulations of $500$ simulated sample draws. QT and QN indicate the results of the quantile ITTTA and ITTNA, and the sum of QTs and QNs represents the sum of ABias and MSE over seven estimates over percentiles. OLST and OLSN indicate the ITTTA and ITTNA OLS estimation results, assuming that D is exogenous, IV indicates the ITTTA IV estimation results, and $T$ is an excluded instrument for $D$. I omit coverages for OLS estimates because none of the results give positive coverage rates due to their short confidence intervals.}
\end{minipage}
\end{table} 
\begin{table}[H]
 \centering
 \caption{Estimator's finite sample performance with $150$ clusters}

 \begin{tabular}{lccccccc}
\hline
& (1) & (2) & (3) & (4) & (5) & & \\
Parameer & ITTTA & ITTNA & QT & QN &  IV &  & \\ \cline{2-6}
Statistic & \multicolumn{5}{c}{95\% coverage probability} & & \\
\hline
$\phi = 0.113$ & 0.956 & 0.964 & 0.968 & 0.966 & 0.024 &  & \\
$\phi = 0.203$ & 0.964 & 0.952 & 0.968 & 0.964 & 0.012 &  & \\
$\phi = 0.338$ & 0.960 & 0.950 & 0.960 & 0.950 & 0.002 &  & \\
\hline
& (1) & (2) & (3) & (4) & (5) & (6) & (7) \\
Parameter & ITTTA & ITTNA & sum of QTs & sum of QNs &  IV & OLST & OLSN \\ \cline{2-8}
Statistic & \multicolumn{7}{c}{Mean absolute bias} \\
\hline
$\phi = 0.113$ & 0.056 & 0.051 & 0.582 & 0.524 & 0.245 & 0.351 & 0.07\\
$\phi = 0.203$ & 0.056 & 0.051 & 0.553 & 0.582 & 0.294 & 0.499 & 0.103\\
$\phi = 0.338$ & 0.056 & 0.051 & 0.537 & 0.964 & 0.394 & 0.800 & 0.197\\ \cline{2-8}
Statistic & \multicolumn{7}{c}{Mean squared error} \\ 
\hline
$\phi = 0.113$ & 0.005 & 0.004 & 0.160 & 0.076 & 0.074 & 0.133 & 0.008\\
$\phi = 0.203$ & 0.005 & 0.004 & 0.082 & 0.234 & 0.101 & 0.259 & 0.015\\
$\phi = 0.338$ & 0.005 & 0.004 & 0.064 & 1.397 & 0.171 & 0.650 & 0.044\\
\hline
 \end{tabular}
\begin{minipage}{325pt}
{\flushleft \fontsize{9pt}{9pt}\selectfont \smallskip NOTE: Every sample is clustered data: $150$ clusters with $10$ observations each. I computed 95\% coverage, mean absolute bias (ABias), and mean squared errors (MSE) using Monte Carlo simulations of $500$ simulated sample draws. QT and QN indicate the results of the quantile ITTTA and ITTNA, and the sum of QTs and QNs represents the sum of ABias and MSE over seven estimates over percentiles. OLST and OLSN indicate the ITTTA and ITTNA OLS estimation results, assuming that D is exogenous, IV indicates the ITTTA IV estimation results, and $T$ is an excluded instrument for $D$. I omit coverages for OLS estimates because none of the results give positive coverage rates due to their short confidence intervals.}
\end{minipage}
\end{table} 

For all the mean and quantile estimate results, close to $95$\% coverages are achieved. In the smaller sample of $750$ observations with $75$ clusters, some quantile estimates fail to achieve $95$\% coverage, and some appear to be slightly conservative. Nevertheless, within the $1,500$ observations with $150$ clusters, which is a little smaller than for the main application, all the mean and quantile estimates achieve $95$\% coverage and only have moderate over-coverage. On average, the proposed estimators' coverage becomes slightly worse when the number of clusters is smaller, and endogeneity is stronger. Nevertheless, endogeneity's strength is sufficiently extreme to make a case for \textit{strong} endogeneity. The IV and OLS estimates fail to provide any valid coverage and their estimate precision is unsatisfactory for both the mean absolute bias and the mean squared error measures.

\subsection{Finite sample performance with a random sample relative to the existing implementation}

\citet{AtheyImbens06} (see Section 5.1 of \citealp{AtheyImbens06}) propose a parametric implementation of their identified parameter, including a vector of the covariates. In their specification, the potential outcomes $Y(d,t)$ and proxy variable $Y_b$ are specified with structural functions, $h_{d,t}, h_b:([0,1],\mathcal{W}) \rightarrow \mathbb{R}$, such that
\[
 Y(d,t) = h_{d,t}(U,W), Y_b = h_b(U,W)
\]
with a common scalar latent variable $U$, resulting in $U \indep W$ given $D$. The key restriction is the additivity of the covariates. Specifically, they consider
\[
 h_{d,t}(u,w) = h_{d,t}(u) + w'\beta,  h_b(u,w) = h_b(u) + w'\beta.
\]
While their consideration accepts the nonparametric specification of $w'\beta$ or $\beta$ being different across $(d,t)$ and the proxy, the above specification is their main implementation. Regarding the recent \textit{did} R package, \cite{did_2020} implemented this specification to allow for the covariates.
The key restriction is the additivity, ensuring that the index of the covariates $W$ does not interact with the latent rank $U$ and that $W$ is independent of $U$. This can be violated in many specifications when the proxy variable is not an exact repeated measure of the outcome measures, and the latent rank is only conditionally similar in rank, not unconditionally through a correlation with $W$.

To demonstrate the differences, I consider the following modified data-generating process: most of the specifications are the same as in Appendix E.1, but there are a few differences. First, as cluster shock is absent, the data represent a random sample of observations. This modification was performed to accommodate \citeauthor{AtheyImbens06}' (\citeyear{AtheyImbens06}) original implementation because they did not incorporate a clustered sample. Second, I consider two cases. For the first, I use a dependent $W$, allowing $W_{c1}$ and $W_{c2}$ to be correlated with the underlying latent variables. Specifically, I set
\[
 W_{c1} = (0.5 U + 0.3 \xi_c)^2
\]
and
\[
 W_{c2} = U_0
\]
where $\xi_c \sim U[0,1]$, when the dependent $W$ is considered. For the second, I left-censor the outcome measures, while the proxy variable is continuously distributed. In particular, I let
\[
 Y^*(d,t) = Y(d,t) 1\left\{Y(d,t) > 0\right\}
\]
for every $(d,t) \in \{(1,1),(1,0),(0,0)\}$. As implementation of the \textit{did} package implementation does not exactly accept censored data, I added a small tiebreak noise to the outcome measures in the estimation. This standard procedure is often undertaken in related implementations. For example, \citeauthor{Rothe_2012}'s (\citeyear{Rothe_2012}) numerical implementation in his proposed estimator involved a similar tiebreak feature to deal with discreteness. Both modifications in the original data-generating process violate \citeauthor{AtheyImbens06}' (\citeyear{AtheyImbens06}) additive separable implementation, while my proposed estimator remains valid, as it allows for a flexible interaction between $U$ and $W$ as well as dependency in (unconditional) $U$ and $W$. I use the same values of $\phi \in \{0.113, 0.203, 0.338\}$ as in Appendix E.1. For reference, Table 7 shows their ITTTA and ITTNA values in large samples.

\begin{table}[H]
\centering
\caption{ITTTA and ITTNA values in large samples with censoring and dependent $W$}
 \begin{tabular}{lcccccc}
\hline
 Strength of endogeneity & mild & medium & strong & mild & medium & strong\\
 $\phi$ & $0.113$ & $0.203$ & $0.338$ & $0.113$ & $0.203$ & $0.338$ \\
\hline
 ITTTA & $0.47$ & $0.43$ & $0.36$ & $0.44$ & $0.41$ & $0.34$ \\
 ITTNA & $0.37$ & $0.41$ & $0.48$ & $0.39$ & $0.42$ & $0.48$\\
  \hline
 Censoring & Y & Y & Y & Y & Y & Y\\
 Dependent $W$ & & & & Y & Y & Y\\ 
  \hline
\end{tabular}
\end{table}

Tables 8 through 10 present the simulation results. Table 8 reports the coverages of the proposed estimators and the corresponding \citeauthor{AtheyImbens06}' (\citeyear{AtheyImbens06}) estimator. The proposed mean estimators are valid throughout the specifications with a slight indication of over-coverage. The quantile estimates also reveal satisfactory results, while the quantile estimates are more unstable relative to the mean estimates in terms of coverage. Nevertheless, \citeauthor{AtheyImbens06}' (\citeyear{AtheyImbens06}) ITTTA estimator fails to achieve coverage in all cases. In particular, the unconditional rank similarity is violated with the censoring results in their confidence intervals, which are very brief. Tables 9 and 10 report the estimators' mean absolute values and mean squared errors. \citeauthor{AtheyImbens06}' (\citeyear{AtheyImbens06}) mean estimates have larger mean absolute biases and mean squared errors when both conditions are violated. Nevertheless, they perform better, although at the cost of failed coverage, when $W$ is independent of the unconditional $U$. This pattern is more apparent in the quantile estimates; however, \citeauthor{AtheyImbens06}' (\citeyear{AtheyImbens06}) quantile estimates fail to cover the target parameter. Nonetheless, as their estimator works well for the data-generating process with the additive separable mean function, and the independence of $W$ and $U$, it would be better to use their estimator for the repeated outcome CiC, as originally proposed. Conversely, the results reveal that we should use my proposed estimator for the other possibilities, especially for the nonrepeated outcomes with discrete outcomes and a continuous proxy.
\begin{table}[H]
 \centering
  \caption{Coverage comparisons of estimators with a $1,000$ sample size} 

 \begin{tabular}{lccccccc}
\hline 
 &               &          & (1) & (2) & (3) & (4) & (5)\\
Parameter & & & ITTNA & ITTTA & QN & QT & ITTTA (AI)\\
Statistic & & & \multicolumn{5}{c}{95\% coverage probability} \\ \cline{4-8}
DGP & Dependent $W$ & Censoring &&&&&\\
\hline
$\phi = 0.113$ &   & Y & 0.964 & 0.960 & 0.944 & 0.924 & 0.770\\
$\phi = 0.203$ &   & Y & 0.970 & 0.958 & 0.970 & 0.928 & 0.830\\
$\phi = 0.338$ &   & Y & 0.968 & 0.970 & 0.966 & 0.984 & 0.878\\
$\phi = 0.113$ & Y & Y & 0.966 & 0.974 & 0.948 & 1.000 & 0.446\\
$\phi = 0.203$ & Y & Y & 0.960 & 0.972 & 0.952 & 0.992 & 0.386\\
$\phi = 0.338$ & Y & Y & 0.956 & 0.974 & 0.964 & 0.980 & 0.478\\
\hline
 \end{tabular}
\begin{minipage}{325pt}
{\flushleft \fontsize{9pt}{9pt}\selectfont \smallskip NOTE: Every sample is a random sample of $1000$ observations. I computed the $95$\% coverage rates using the Monte Carlo simulations of $500$ simulated sample draws. QT and QN indicate the quantile ITTTA and ITTNA results. The ITTTA (AI) represents \citeauthor{AtheyImbens06}' (\citeyear{AtheyImbens06}) ITTTA estimate. The OLS, IV, and \citeauthor{AtheyImbens06}' (\citeyear{AtheyImbens06}) QT estimate coverages were omitted because none of the results provide positive coverage rates due to their brief confidence intervals.}
\end{minipage}
\end{table} 

\begin{table}[H]
 \centering
  \caption{MSEs and Abias comparisons of estimators with a $1,000$ sample size} 
 \begin{tabular}{lccccccccc}
\hline
 &          &               & (1) & (2) & (3) & (4) & (5) & (6)\\
Parameter & & &  ITTTA & ITTNA & ITTTA (AI) &  IV & OLST & OLSN\\
Statistic & & &  \multicolumn{6}{c}{Mean absolute bias} \\ \cline{4-9}
DGP & Dependent $W$ & Censoring &&&&&&\\
\hline
$\phi = 0.113$ &  & Y & 0.054 & 0.051  & 0.051 & 0.371 & 0.460 & 0.072\\
$\phi = 0.203$ &  & Y & 0.054 & 0.052  & 0.045 & 0.410 & 0.581 & 0.095\\
$\phi = 0.338$ &  & Y & 0.052 & 0.053  & 0.038 & 0.475 & 0.839 & 0.174\\
$\phi = 0.113$ & Y & Y & 0.037 & 0.036  & 0.083 & 0.389 & 0.385 & 0.035\\
$\phi = 0.203$ & Y & Y & 0.035 & 0.037  & 0.084 & 0.415 & 0.417 & 0.036\\
$\phi = 0.338$ & Y & Y & 0.035 & 0.039  & 0.069 & 0.476 & 0.476 & 0.037\\
\hline
Statistic & & & \multicolumn{6}{c}{Mean squared error} \\ \cline{4-9}
DGP & Dependent $W$ & Censoring &&&&&&\\
\hline
$\phi = 0.113$ &  & Y & 0.004 & 0.004  & 0.003 & 0.153 & 0.223 & 0.008\\
$\phi = 0.203$ &  & Y & 0.004 & 0.004  & 0.003 & 0.184 & 0.349 & 0.013\\
$\phi = 0.338$ &  & Y & 0.004 & 0.004  & 0.002 & 0.243 & 0.715 & 0.036\\
$\phi = 0.113$ & Y & Y & 0.002 & 0.002  & 0.008 & 0.157 & 0.151 & 0.002\\
$\phi = 0.203$ & Y & Y & 0.002 & 0.002  & 0.008 & 0.178 & 0.177 & 0.002\\
$\phi = 0.338$ & Y & Y & 0.002 & 0.002  & 0.006 & 0.232 & 0.230 & 0.002\\
\hline
 \end{tabular}
\begin{minipage}{325pt}
{\flushleft \fontsize{9pt}{9pt}\selectfont \smallskip NOTE: Every sample is a random sample of $1000$ observations. I computed the mean absolute bias and mean squared errors using the Monte Carlo simulations of $500$ simulated sample draws. QT and QN indicate the quantile ITTTA and ITTNA results, and sum of QTs and QNs are the sum of ABias and MSE over seven estimates over percentiles. The ITTTA (AI) represents \citeauthor{AtheyImbens06}' (\citeyear{AtheyImbens06}) ITTTA estimate. The sum of QTs (AI) represents the sum of ABias and MSE over seven estimates over percentiles for \citeauthor{AtheyImbens06}' (\citeyear{AtheyImbens06}) quantile estimates. OLST and OLSN indicate the results of OLS estimation of the ITTTA and the ITTNA assuming that $D$ is exogenous, and IV indicates the results of IV estimation of the ITTTA assuming that $T$ is an excluded instrument for $D$.}
\end{minipage}
\end{table}

\begin{table}[H]
 \centering
  \caption{MSEs and Abias comparisons of estimators with a $1,000$ sample size, continued} 
 \begin{tabular}{lccccc}
\hline
 &              &          & (1) & (2) & (3)\\
Parameter & & & sum of QTs & sum of QNs & sum of QTs (AI)\\
Statistic & & & \multicolumn{3}{c}{Mean absolute bias} \\ \cline{4-6}
DGP & Dependent $W$ & Censoring &&\\
\hline
$\phi = 0.113$ &  & Y & 0.763 & 0.704 & 0.243\\
$\phi = 0.203$ &  & Y & 0.657 & 0.778 & 0.225\\
$\phi = 0.338$ &  & Y & 0.572 & 1.160  & 0.206\\
$\phi = 0.113$ & Y & Y & 0.503 & 0.512  & 0.411\\
$\phi = 0.203$ & Y & Y & 0.464 & 0.622  & 0.368\\
$\phi = 0.338$ & Y & Y & 0.456 & 1.085  & 0.344\\
\hline
Statistic & & & \multicolumn{3}{c}{Mean squared error} \\ \cline{4-6}
DGP & Dependent $W$ & Censoring &&\\
$\phi = 0.113$ &  & Y &  0.415 & 0.253  & 0.030\\
$\phi = 0.203$ &  & Y &  0.173 & 0.474  & 0.027\\
$\phi = 0.338$ &  & Y &  0.080 & 1.681  & 0.025\\
$\phi = 0.113$ & Y & Y &  0.176 & 0.104 & 0.08\\
$\phi = 0.203$ & Y & Y &  0.127 & 0.317 & 0.062\\
$\phi = 0.338$ & Y & Y &  0.158 & 1.620  & 0.052\\
\hline
 \end{tabular}
\begin{minipage}{325pt}
{\flushleft \fontsize{9pt}{9pt}\selectfont \smallskip NOTE: Every sample is a random sample of $1000$ observations. I computed the mean absolute bias and mean squared errors using the Monte Carlo simulations of $500$ simulated sample draws. QT and QN indicate the quantile ITTTA and ITTNA results, and sum of QTs and QNs are the sum of ABias and MSE over seven estimates over percentiles. The ITTTA (AI) represents \citeauthor{AtheyImbens06}' (\citeyear{AtheyImbens06}) ITTTA estimate. The sum of QTs (AI) represents the sum of ABias and MSE over seven estimates over percentiles for \citeauthor{AtheyImbens06}' (\citeyear{AtheyImbens06}) quantile estimates. OLST and OLSN indicate the results of OLS estimation of the ITTTA and the ITTNA assuming that $D$ is exogenous, and IV indicates the results of IV estimation of the ITTTA assuming that $T$ is an excluded instrument for $D$.}
\end{minipage}
\end{table}

\section{Application to a two-sided noncompliance experiment}

In the main text, I apply the proposed strategy to the one-sided noncompliance microcredit experiment. As discussed in Section 6, the justification is highly complex for two-sided noncompliance. Nevertheless, I demonstrate that the analogue procedure is valid for two-sided noncompliance cases when taking up the treatment in the control group is \textit{stable}. Hence, I offer an additional application to the proposed strategy for two-sided noncompliance experiments. On that note, \cite{Gertler_Martinez_Rubio-Codina_2012} studied the effect of a cash transfer program in Mexico, called the Oportunidades program. They focused on its long-term impact on living standards, possibly through productive asset investments. A form of social assistance, the Oportunidades program was a large-scale cash transfer experiment based on the well-known Progresa program. Specifically, the program distributed cash payments to randomly selected poor households, conditional on passing eligibility requirements such as undertaking health checkups and sending children to school. \cite{Gertler_Martinez_Rubio-Codina_2012} found that households receiving 18 months of cash transfers invested in productive assets such as livestock relative to control group households, which had to wait 18 months to receive payments. They also showed that the impact of the cash transfer could sustain even four years later. They argued that investing in productive assets is the main pathway to achieve this long-lasting effect by revealing the absence of other channels such as other income, formal/informal credit, and health.

In this section, I directly estimate the ITTTA and ITTNA to see if the ownership of productive livestock is indeed the main channel. If this ownership is the main channel, then the ITTTA should have a large magnitude in the same direction as the ITT, and the ITTNA should be small in absolute value. $T$ represents the cash transfer's randomized assignment and $D$ is the treatment indicator (productive livestock ownership). The cash transfer program encourages households to invest more, possibly allowing them to own productive livestock; however, nothing bars them from owning livestock prior to the assignment. Therefore, it is a two-sided noncompliance experiment. Furthermore, $Y$ is the value in pesos for 2003 consumption, that is, four years after the program, and as my strategy requires a proxy variable; in this case, the same consumption measure in the baseline period is a good proxy given that $Y$ is often continuously distributed with no mass point.

However, to the best of my knowledge, Progresa program did not collect a detailed consumption measure in the 1997 baseline survey (ENCASEH). Alternatively, I use poverty index measures to determine eligibility for the Progresa and Oportunidades programs. As this index measure is a predicted value from household characteristics for the poverty indicator, I assume that the index is a proxy for the unobserved baseline consumption, but there can be a rank-shifting change in the consumption level for the control units over time. After merging the 1997 baseline survey, the final sample contains $9,837$ observations with nonmissing $Y_b$ in $506$ communities. The sample size of my analysis is further reduced to $9,536$ for those who have nonmissing observations of $D$. Table 11 describes the take-up rates of $D$ by assignment status. 
\begin{table}[H]
\centering
\caption{Relative frequencies by assignment and take-up status}
 \begin{tabular}{lcc}
         & $T = 0$ & $T = 1$\\
  \cline{2-3}
  $D = 0$ & 28.02\% & 43.44\%\\
  $D = 1$ & 9.73\% & 18.81\% \\
  \hline
 \end{tabular}
\end{table}
As Table 11 shows, the take-up rate of ownership is roughly $20$\% to $30$\%. The compliance rate is small, at about $4$\%, which is statistically significant at the $5$\% level in the regression analysis. Table 12 summarizes the main results.
\begin{table}[H]
 \centering
  \caption{Estimates with the baseline proxy $Y_b$} 

 \label{tbl:mainResults}
 \begin{tabular}{lccccc}
 \hline
 & (1) & (2) & (3) & (4) & (5) \\
Model & RS & RS & RS & OLS & OLS\\
Parameter & ITTNA & ITTTA & LATE & ITT & ITT\\ \cline{2-6}
Outcome & \multicolumn{5}{c}{level of consumption}\\
\hline
  Assignment ($T$) & $3.49$ & $23.88^{***}$ & $127.70$ & $10.14^{*}$ & $9.97^{**}$ \\
  by subgroups of ($D$)    & ($5.91$) & ($6.52$) & ($388.99$) & ($5.88$) & ($4.49$)\\
  \hline
Full covariates & & & & & Y \\
  Obs       & $9,536$ & $9,536$ & $9,536$  & $9,837$ & $9,837$ \\ 
  \hline
 \end{tabular}
 \vspace{0.3cm}
\begin{minipage}{325pt}
{\flushleft
 \fontsize{9pt}{9pt}\selectfont \smallskip NOTE: Standard
 errors reported in parenthesis are generated from 300 bootstrap draws clustered in
 community levels for (1)-(3). For (4), I use the limited set of covariates as for (1)-(3) in the subsample with nonmissing observations of $Y_b$, and I use the full set of covariates for (5). *,**,*** indicates statistical significance of
 10\%,5\% and 1\% sizes respectively. Logit link is used for (1)-(3).
 } 
\end{minipage}
\end{table} 

Column (1) presents the ITTNA estimate, representing the assignment's direct effect estimate for those who do not own productive animals. Column (2) shows the ITTTA estimate, indicating the combined impact of the assignment's direct effect and that of owning productive animals for the sample that does. Column (3) represents the LATE-analogue estimate, discussed in Section 6 and Appendix C.2, which is imprecise because of the low compliance rate. 

For the rank similarity estimates in columns (1) through (3), I use the subset of covariates employed in the original regression analysis, reported in column (4), as the ITT estimate based on OLS. The original specification includes the head of household's age, spousal age, age squared and baseline education dummies, ethnicity of household head, as indicated by language, baseline household size, baseline dummies for household demographics, baseline assets, and baseline community characteristics. By contrast, I limit the covariates to the baseline amounts of assets, including number of animals and extent of total operational land holdings in hectares, head of household's age, and limited baseline dummies of household characteristics such as agricultural homeownership, big farm ownership, and gender and education level of household head. The reason for the limitation is that most of the covariates used in this application are binary indicators. As such, incorporating high-dimensional discrete variables is an important topic for future research. Column (5) is the ITT estimate based on OLS using the full set of covariates for a reference under the same matched subsample with the baseline survey as used for column (4).

Table 12 results that the ITTNA is small and statistically insignificant, while the ITTTA is large and statistically significant. Mean consumption in the control group in this sample is $188.1$. Therefore, the ITTNA is only $1$\% to $2$\% of the mean. Conversely, the ITTTA reveals that the assignment increases long-term consumption by more than $12$\% for those who own productive assets. Therefore, the impact of the assignment is primarily due to the investment channel, while the other paths are unlikely to drive this huge impact.

Considering the identification of LATE, it is arguable that the homogeneity assumption may be reasonable in this application. The following formula,
\[
 E[DE(1)|T = 1, D(1) = 1] = E[Y(1,1) - Y(0,1)|T = 1, D(1) =1],
\]
represents the direct effect of the cash transfer holding ownership, which remains unchanged for those who would have owned an animal if they had received the transfer. Here, $Y(0,1)$ is the counterfactual outcome that ensures when they own an animal without receiving the cash transfer. Thus, $Y(1,1)$ may be associated with the smaller amount of remaining cash, attributed to purchasing the animal; however, $Y(0,1)$ also corresponds to the cash amount reduced by the expenditure of owning the animal without receiving the transfer. Therefore, the net increase in amount from the cash transfer for takers is the same as the remaining cash transfer value for nontakers. For the latter group, I compare the outcomes of not owning the animal by considering the full cash transfer amount spent on other channels:
\[
 E[DE(0)|T = 1, D(1) = 0] = E[Y(1,0) - Y(0,0)|T = 1, D(1) =0].
\]
Thus, both direct effects capture the effect of additional cash amounts with the same net cash values. If the impact of the additional cash amount is additive in the long-term consumption level, then the homogeneity assumption may be justified. Nevertheless, as the complier rate is small, it is not feasible to obtain a precise LATE estimate, even though the homogeneity assumption may be reasonable.







\bibliographystyle{chicago}

\bibliography{ref20170301}
\end{document}